\documentclass[
submission
 , nomarks
]{dmtcs-episciences}

\usepackage[utf8]{inputenc}
\usepackage{subfigure}

\usepackage[round]{natbib}

\author{Santiago Viertel\thanks{Supported by Coordination for the Improvement of Higher Education Personnel (CAPES) proc. 1431510/2014-9.}
  \and Andr\'e Lu\'is Vignatti}
\title[Labeling Algorithm and Compact Routing Scheme for a Small World Network Model]{Labeling Algorithm and Compact Routing\\Scheme for a Small World Network Model}
\affiliation{
  DINF, Federal University of Paran\'a, Brazil}
\usepackage{amsmath}
\usepackage[linesnumbered,vlined,ruled,commentsnumbered]{algorithm2e}
\usepackage{color}
\usepackage{multicol}
\usepackage{comment}
\newtheorem{theorem}{Theorem}
\newtheorem{lemma}[theorem]{Lemma}
\newtheorem{corollary}[theorem]{Corollary}
\newtheorem{definition}{Definition}
\newtheorem{fact}{Fact}

\DeclareMathOperator*{\E}{\textrm{E}}
\SetKwInOut{Instance}{Instance}
\SetKwInOut{Returns}{Return}
\SetKwFor{ForEach}{for each}{do}

\SetKwProg{Alg}{Algorithm}{}{}
\SetKwFunction{UTSW}{UTSW}
\SetKwFunction{CS}{C4S}
\SetKwFunction{RCS}{RC4S}
\SetKwFunction{DRA}{DRA}
\SetKwFunction{LPRA}{LPRA}
\SetKwFunction{LRERA}{LRERA}
\SetKwFunction{RSLA}{RSLA}
\SetKwFunction{RSVFA}{RSVFA}
\SetKwFunction{ACLA}{ACLA}
\SetKwFunction{LA}{LA}
\newcommand{\nil}{\texttt{nil}}
\newcommand{\true}{\texttt{true}}
\newcommand{\false}{\texttt{false}}

\begin{document}
\publicationdetails{VOL}{2019}{ISS}{NUM}{SUBM}
\maketitle

\begin{abstract}

This paper defines the toroidal small world labeling problem that asks for a labeling of the vertices of a network such that the labels possess information that allows a compact routing scheme in the network.
We consider the problem over a small world network model we propose.
Both the model and the compact routing scheme have applications in peer-to-peer networks.
The proposed model is based on the model of Kleinberg (2000), and generates an undirected two-dimensional torus with one random long-range edge per vertex.
These random edges create forbidden cycles that mimic the underlying torus topology, and this behavior confuses attempts for extracting the routing information from the network.
We show that such forbidden cycles happen with small probability, allowing us to use a breadth-first search that finds the vertices positions on the torus.
The positions are pairs of integer numbers that provides routing information to a greedy routing algorithm that finds small paths of the network.
We present a linear time labeling algorithm that detects and removes the random edges, finds the underlying torus and labels almost all vertices through a breadth-first search.
The labeling algorithm is then used by a compact routing scheme for the proposed small world model.

\end{abstract}

\keywords{small world network, labeling algorithm, compact routing scheme, greedy routing}
\section{Introduction}
\label{sec:introduction}

Milgram's experiment~\citep{Milgram:1967} concludes that the average path length between two vertices in real-world networks is small.
This phenomenon is known as ``small world'', which motivates the development of formal models.
There are two well known small world graph models~\citep{Watts:1998,Kleinberg:2000a}, both focusing on small average path length and also vertices clustering.
Small world graphs have many applications such as networks of electronic mail messages exchanging~\citep{Adamic:2005}, friendship~\citep*{Liben-Nowell:2005}, Internet domain~\citep*{Krioukov:2004}, peer-to-peer~\citep*{Jovanovic:2001,Manku:2004} and wireless sensors~\citep*{Liu:2009}.

\cite{Kleinberg:2000a} presents a small world graph model, and a greedy routing algorithm for his model, and proves that the algorithm delivers messages with few forwards.
The model consists of a two-dimensional $n\times n$ lattice with some random directed long-range edges.
The greedy routing algorithm forwards a received message to the neighbor closest to the target.
Kleinberg proves that this decentralized routing procedure executes O$\left(\log^2n\right)$ expected number of message forwardings in graphs generated by his model.
Some works make modifications on Kleinberg model and on the greedy routing in order to obtain improvements or trade-offs between storage needed in bits per vertex and paths lengths \citep*{Martel:2004,Zeng:2005,Liu:2009}.
Kleinberg routing algorithm requires labels assigned to the vertices and routing tables with positioning information.
Finding these labels is the main goal of this paper.

In the context of compact routing, we refer to \emph{preprocessing algorithms} those that generate labels and routing tables to each vertex, which typically stores the graph topological information.
A \emph{routing algorithm} uses the label and routing table of a vertex to forward messages through it \citep{Cowen:1999}.
A \emph{routing scheme} is a combination of a preprocessing algorithm and the related routing algorithm, that provides a complete mechanism for routing messages in a network \citep{Thorup:2001a}.
A routing scheme is \emph{compact} if each vertex uses sublinear storage space in the number of vertices \citep*{Chen:2009,Chen:2012}.
The worst and expected cases on the number of bits for storage per vertex are frequent metrics for measure effectiveness of routing schemes \citep{Thorup:2001a,Chen:2012}.
Some works also measure through \emph{stretch} \citep*{Cowen:1999,Abraham:2006a,Chen:2009}, which is the maximal ratio between the path lengths performed by the routing scheme and a shortest path over all pairs of vertices.
Others measure through the \emph{expected path length} performed by the routing scheme \citep*{Zeng:2005,Zeng:2006,Liu:2009}.

There are compact routing schemes for general graphs~\citep{Cowen:1999,Thorup:2001a}, trees~\citep{Thorup:2001a}, graphs with low doubling dimension~\citep*{Abraham:2006a,Konjevod:2007a} and power law graphs~\citep{Chen:2009,Chen:2012}.
Some works deal with greedy routing on metric spaces~\citep*{Ban:2010,Bringmann:2017}, even in dynamic networks~\citep*{Krioukov:2008,Krioukov:2009}.
Some of them propose the study of embeddings in metric spaces.
Embedding graphs in metric spaces defined by the two-dimensional lattice is an interesting problem in the context of greedy routing in graphs generated by the Kleinberg model.
An embedding attaches to each vertex a label with a pair numbers representing the vertex position in the lattice.
\cite{Kleinberg:2006} also mentioned this problem.

Labeling the vertices with their respective positions in the lattice is a problem that rises from specific contexts.
In peer-to-peer file sharing networks, such as Symphony~\citep*{Manku:2003} and Freenet~\citep*{Clarke:2000}, the algorithm for searching documents in the network should find as many as possible source nodes with a fell hops.
A solution is to create a peer-to-peer overlay network with a small world model, which is an alternative to hypercubes, for example.
In this context, hypercubes have diameter $\log|V|$, for a graph with $|V|$ vertices, and all vertices have the same degree $\log|V|$.
On the other hand, small world graphs generated by Kleinberg model, for example, have expected diameter $\log|V|$ while all vertices have at most constant outdegree.
Constant degree for all vertices is an important feature that directly impacts in the time complexity of distributed search and routing algorithms.
We present a small world model that generates random undirected graphs with expected diameter $\log|V|$, expected constant degree in all vertices and with diameter $|V|^{1/2}$ in the worst case, which is smaller than $2|V|^{1/2}$ of the Kleinberg model in the two-dimensional lattice.

The need of routing messages through small paths may naturally become apparent in large peer-to-peer networks that are already created and in using.
In this context, the nodes may not know their positions in the initial regular structure generated by the small world model, as the document search algorithm may not use this information.
Thus, an algorithm that retrieves these positions is necessary for routing messages through small paths in these peer-to-peer networks.
\cite{Sandberg:2006} claims that a vertex may not have stored the positions of its neighbors in some applications.
He solved the problem of finding labels to the vertices through a Markov chain Monte Carlo algorithm that runs in a distributed fashion.
His algorithm iteratively estimates the positions in the lattice, assuming that the network was created by the Kleinberg model.

We present an alternative algorithm that runs in expected linear time on the size of the network, and can easily be adapted to a distributed version that runs in expected constant time in each node.
The algorithm labels almost all vertices of small world graphs built upon a torus, instead of a lattice.
It tries to identify the spanning torus and labels the vertices with their positions in the torus.
The \emph{spanning torus} is a torus subgraph generated by the small world model presented in this paper.
It is a spanning graph, by the fact that it comprehends all the vertices.
Section~\ref{sec:undirectedToroidalSmallWorldModel} presents it with more details.
The algorithm runs a bounded depth search from each vertex aiming to find all the four-cycles rooted in it.
The vertex recognizes lattice patterns in all combinations of four-cycles and identifies the edges in the torus, if possible.
Finally, the algorithm labels all vertices with identified edges and their neighbors by a global breadth-first search.

\textbf{Contributions.} We present a small world model based on Kleinberg model, called \emph{undirected toroidal small world} (\UTSW) model.
It differs on the generation of graphs over a torus, instead of a lattice, and with undirected long-range edges, instead of directed.
We present upper and lower bounds on the normalizing factor of $\Theta\left(\log^{-1}n\right)$, where $n$ is a parameter of the model that defines the torus size.
The \emph{normalizing factor} is a multiplicative factor that keeps the probability distribution summing to one on the long-range edges random generation.
We present a formal definition for the labeling problem, an upper bound on the probability of existence of four-cycles composed by edges not in the spanning torus and show that this probability is small, which is O$\left(\log^{-1}n\right)$.
We also present an expected linear time algorithm that labels almost all vertices and show a compact routing scheme for \UTSW model.

\textbf{Organization.} Section~\ref{sec:relatedWorks} presents related works on small world graph models, improvements and trade-offs in greedy routing and labeling algorithms for routing.
Section~\ref{sec:undirectedToroidalSmallWorldModel} presents our model and an analysis of its normalizing factor.
Section~\ref{sec:toroidalSmallWorldLabelingProblem} presents the formal definition of the \emph{toroidal small world labeling problem}.
Section~\ref{sec:cyclesOfSize4OutsideTheTorus} presents a sequence of results in order to bound the probability of the existence of four-cycles, rooted in a given vertex, that do not belong to the spanning torus.
This upper bound is important for the probabilistic analysis of the labeling algorithm.
Section~\ref{sec:labelingAlgorithm} presents all procedures of the labeling algorithm, together with their analyses, and finishes presenting a compact routing scheme.
Section~\ref{sec:conclusion} presents the conclusion and future works.
\section{Related Works}
\label{sec:relatedWorks}

\cite{Watts:1998} claim that many technological and social networks have topology between regular and random.
They present a model that generates an undirected $n$-vertex ring and connects each vertex with its $k$ closest vertices.
After that, it changes the head of each edge uniformly and independently at random with probability $c$.
The resulting graphs have small average of path length for some values of $c$, as \cite{Bollobas:1988} prove that a cycle with a random matching has diameter $\Theta(\log n)$ with high probability.

\cite{Kleinberg:2000a} claims that people have ``close'' and ``far'' contacts.
He presents a model based on geographic positions.
It generates a $n\times n$ lattice where each vertex has a distinct position in the lattice, represented by a pair in $\{1,\ldots,n\}\times\{1,\ldots,n\}$.
For each vertex, it creates directed edges to vertices within $p\ge1$ steps in the lattice and $q\ge0$ directed edges with independent random trials.
Each edge with tail in $u\in V$ has head $v\in V\setminus\{u\}$ with probability $Z_u\cdot d_{uv}^{-r}$, where $d_{uv}$ is the distance in the lattice between $u$ and $v$, $r\ge0$ is the parameter for the probability distribution and $Z_u$ is the normalizing factor in $u$.
The \emph{distance in the lattice} between $u$ and $v$ is $d_{uv}=|x_u-x_v|+|y_u-y_v|$, where $(x_w,y_w)$ is the position of $w\in V$ in the lattice.
The \emph{normalizing factor} in $u$ is $Z_u=\left(\sum_{w\in V\setminus\{u\}}d_{uw}^{-r}\right)^{-1}$.
The probability distribution $Z_u\cdot d_{uv}^{-r}$ is referred in this paper as \emph{inverse $r^\textrm{th}$-power distribution}.

Kleinberg also investigates the influence of $r$ in a specific greedy routing procedure.
\emph{Myopic search}\footnote{As the author named it in one of his books \citep{Easley:2010}.} is a greedy routing procedure that uses the positions of the neighbors and the target.
In myopic search, a vertex $u$ forwards the message to its neighbor $\textrm{argmin}_{v\in N(u)}d_{vt}$, where $N(u)$ is the set of neighbors of $u$ and $t\in V\setminus\{u\}$ is the target.
He proves that myopic search performs O$\left(\log^2n\right)$ expected number of forwards, for $p=q=1$ and $r=2$.
\cite{Martel:2004} show that this bound is tight.
Myopic search is considered an efficient algorithm because a polylogarithmic function in $n$ defines the expected number of forwards.
\cite{Sandberg:2006} presents the problem of vertices labeling when the positions are unknown.
His algorithm estimates the position of each vertex in the case of the graph is toroidal and created by the Kleinberg model.
The labeling algorithm presented in this paper labels a large fraction of vertices with exact positions, which distinguishes it from the statistical estimations that the Sandberg method outputs.

Analyzing the greedy routing in alternative models are also considered.
\cite{Manku:2004} present a routing algorithm similar to myopic search that also considers the neighbors of the neighbors.
The expected number of forwards is O$\left(\frac{\log^2n}{q\cdot\log q}\right)$ for a unidirectional $n$-ring based model, where $q$ is the number of directed edges generated by random trials.
\cite{Martel:2004} present a routing algorithm that also considers the positions of the long-range edges endpoints of the $\log n$ nearest vertices.
The expected number of forwards is O$\left(\log^{1+1/d}n\right)$ for a $d$-dimensional lattice with long-range edges created with the inverse $d^\textrm{th}$-power distribution.
\cite*{Fraigniaud:2006} obtain this result through an oblivious routing algorithm, that does not change the message header.

\cite{Zeng:2005} present a unidirectional $n$-ring based model.
For each vertex, it creates one long-range edge with the inverse first-power distribution and two augmented links with the vertices within distance $\log^2n$ chosen uniformly at random.
They define a non-oblivious and an oblivious routing algorithms that consider the positions of the O$(\log n)$ nearest vertices reachable through augmented links.
Both have O$(\log n\log\log n)$ expected number of forwards.
\cite{Zeng:2006} generalize the model for $d$ dimensions and define an oblivious algorithm that performs O$(\log n)$ expected number of forwards.

\cite{Liu:2009} present a model that creates one long-range edge per subgroup of $k\times k$ vertices.
The routing algorithm considers the positions of the endpoints of the O$(\log m)$ nearest long-range edges, where $m=n/k$.
The expected number of forwards is O$\left(\log m\cdot\log^2\log m\right)$.
More works~\citep*{Aspnes:2002,Dietzfelbinger:2009} also present models, greedy routing algorithms and their expected number of forwards.

In the routing context, each vertex usually has a label and a routing table that encode network topological information.
\emph{Preprocessing algorithms} take as input a graph that models the network and output the labels and routing tables of all vertices.
\emph{Labeling algorithms} assign addresses, or labels, to the vertices \citep{Cowen:1999}.
Running time and labels size are the main metrics for analyzing labeling algorithms.

\cite{Cowen:1999} presents a compact routing scheme for general networks.
The labeling algorithm runs in \~O$\left(n^{1+2\alpha}+n^{1-\alpha}m+n^{(1+\alpha)/2}m\right)$ time\footnote{The asymptotic notation \~O hides polylogarithmic multiplicative factors from the bound.}, where $n$ is the number of vertices, $m$ is the number of edges and $0<\alpha<1$.
It chooses a set of vertices called landmarks and generates labels with $2\lfloor\log n\rfloor+\lfloor\log\Delta\rfloor+3$ bits, where $\Delta$ is the maximum degree of the graph.
\cite{Thorup:2001a} present a landmark selection process that changes the running time to \~O$(nm)$.

There are specialized compact routing schemes.
\cite{Abraham:2006a} present one for networks with low doubling dimension.
The labeling algorithm has linear running time and it generates labels with $\lceil\log n\rceil$ bits.
\cite{Chen:2009,Chen:2012} present a compact routing scheme for power law networks.
The labels encode a shortest path and, because of that, they have size of O$(\log n\log\log n)$ bits with probability $1-$o$(1)$.
The preprocessing algorithm runs in O$\left(n^{1+\gamma}\log n\right)$ expected time, where $\epsilon<\gamma<1/3+\epsilon$.

Some authors use existing compact routing schemes.
\cite{Brady:2006} present a compact routing scheme for power law networks that uses the Thorup and Zwick scheme for trees~\citep{Thorup:2001a}.
The labels have O$\left(e\log^2n\right)$ bits, where $e$ assumes small values.
\cite{Konjevod:2007a} present a compact routing scheme for networks with low doubling dimension that also uses the Thorup and Zwick scheme for trees, where the labels have $\lceil\log n\rceil$ bits.
\section{Undirected Toroidal Small World Model}
\label{sec:undirectedToroidalSmallWorldModel}

This section presents the \emph{undirected toroidal small world} (\UTSW) model.
The model is similar to the \cite{Kleinberg:2000a} model.
Roughly, instead of using a grid, it ``ties the borders'' of the lattice, obtaining a torus.
Also, Kleinberg uses directed long-range edges, but we use undirected instead.
Let $S^2=S\times S$ and $[\![n]\!]=\{0,1,...,n-1\}$.
At first, the model generates a two-dimensional $n\times n$ undirected torus as follows.
It creates $n^2$ vertices such that each vertex is a distinct element of $[\![n]\!]^2$.
After that, each vertex $(i,j)$ is connected with the vertices $(i,(j+1)\mod n)$ and $((i+1)\mod n,j)$ through undirected edges.
Now, we have a torus as Figure~\ref{fig:2DimensionalTorus} shows, in that case, it is the two-dimensional undirected torus with size three.
We denote this generated torus as $T=(V,E')$ \emph{with size $n$} in the rest of the paper, unless the generated torus or $T$ are locally redefined.
Note that $|V|=n^2$ and $|E'|=2|V|$.
We also assume sizes $n\ge3$, otherwise the torus may have parallel edges or loops, both of which we do not consider in this paper.

\begin{figure}
	\centering
	\includegraphics{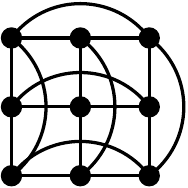}
	\caption{The two-dimensional undirected torus with size $n=3$.}
	\label{fig:2DimensionalTorus}
\end{figure}

Let $d_{uv}$ be the distance between $u$ and $v$ on $T$.
Note that $d_{uv}=\min(|x_u-x_v|,n-|x_u-x_v|)+\min(|y_u-y_v|,n-|y_u-y_v|)$ for all $u,v\in V$, where $(x_w,y_w)$ is defined for $w\in V$ during the torus generation.
After the torus creation, each vertex $u\in V$ chooses another vertex $v\in V\setminus\{u\}$ to create an edge.
The choice is sampled from the inverse second-power distribution, \textit{i.e.}, $u$ chooses $v$ with probability $Z_u\cdot d_{uv}^{-2}$, where $Z_u$ is the \emph{normalizing factor} in $u$.
The \emph{normalizing factor} $Z_u$ is a value that keeps the probability distribution summing to one in each $u$.
Its value is
\[Z_u=\left(\sum\limits_{w\in V\setminus\{u\}}d_{uw}^{-2}\right)^{-1}.\]

To avoid parallel edges, an edge is not created if it has already been created.
We refer the edges created in the torus generation process as \emph{local edges} and those created by the randomized process as \emph{long-range edges}.
Let $C_{uv}$ be the event of vertex $u$ choosing the vertex $v$ to create a long-range edge.
Note that the probability distribution $\Pr\left(C_{uv}\right)=Z_u\cdot d_{uv}^{-2}$, for all $v\in V\setminus\{u\}$, sums to one.
We use the distance function $d$ and the event $C$ in the rest of the paper.
We also assume that the positions $(x_w,y_w)$ of all $w\in V$, computed during the torus generation, are not an output of the model.
In fact, finding these positions is the problem that we solve in this paper.

\begin{theorem}
	$Z_u=Z_v$ for all $u,v\in V$.
	\label{thm:normalizingFactor}
\end{theorem}
\begin{proof}
	Let $P_{ui}=\{w\in V|d_{uw}=i\}$ for all $1\le i\le n$.
	Figure~\ref{fig:VerticesAtDistI} shows the vertices in $P_{u3}$ in a torus with $n=7$.
	As $T$ is symmetrical for all $u$ and $v$, so $|P_{ui}|=|P_{vi}|$ and, then
	\[Z_u=\left(\sum\limits_{w\in V\setminus\{u\}}d_{uw}^{-2}\right)^{-1}=\left(\sum\limits_{i=1}^{n}|P_{ui}|\cdot i^{-2}\right)^{-1}=\left(\sum\limits_{i=1}^{n}|P_{vi}|\cdot i^{-2}\right)^{-1}=Z_v.\]
\end{proof}
\begin{figure}
	\centering
	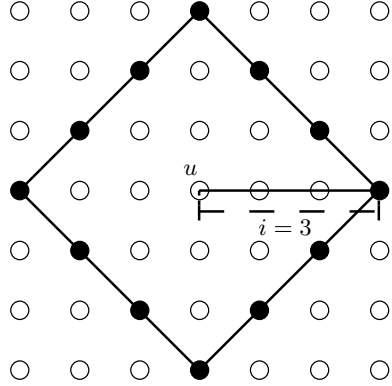
	\caption{Vertices at a distance $i=3<n/2$ from $u$, with $n=7$.}
	\label{fig:VerticesAtDistI}
\end{figure}

Despite that each vertex has its own normalizing factor value, Theorem~\ref{thm:normalizingFactor} shows that all of them are equal to each other.
So, we denote it as $Z$, as shown in Definition~\ref{dfn:normalizingFactor}.
The \UTSW model is defined in Definition~\ref{dfn:UTSWModel}.

\begin{definition}[Normalizing factor]
	The \emph{normalizing factor} is $Z=\left(\sum\limits_{v\in V\setminus\{u\}}d_{uv}^{-2}\right)^{-1}$ for any $u\in V$.
	\label{dfn:normalizingFactor}
\end{definition}

\begin{definition}[\UTSW model]
	The \emph{undirected toroidal small world} (\UTSW) model is a graph $G=(V,E)$ over a two-dimensional undirected torus with size $n\ge3$ and, for all $u\in V$ and a $v\in V\setminus\{u\}$, the edge $\{u,v\}$ is included in $E$ with probability $\Pr\left(C_{uv}\right)=Z\cdot d_{uv}^{-2}$.
	\label{dfn:UTSWModel}
\end{definition}

We assume that $G$ is a \UTSW graph from now on.
The remaining of this section presents some results related to the \UTSW model that are used throughout the paper.

\begin{lemma}
	Let $P_{ui}=\{v\in V|d_{uv}=i\}$ for each $u\in V$ and $1\le i\le n$.
	Then, $|P_{ui}|=4i$ if $i<n/2$, and $|P_{ui}|<4i$ otherwise.
	\label{lem:numberOfVerticesAtDistanceI}
\end{lemma}
\begin{proof}
	The proof follows by ``sweeping'' in one of the two dimensions of $T$.
	When $i<n/2$,
	\[|P_{ui}|=1+\sum\limits_{j=1}^{i-1}2+2+\sum\limits_{j=1}^{i-1}2+1=4+2\cdot2(i-1)=4i.\]
	When $i\ge n/2$, $4i$ is an upper bound for $|P_{ui}|$.
\end{proof}

\begin{fact}[Harmonic number]
	Let $H_k$ be the $k^\textrm{th}$ \emph{harmonic number}, then $\ln(k+1)<H_k\le\ln k+1$.
	\label{fct:harmonicNumber}
\end{fact}

\begin{theorem}
	The normalizing factor $Z<(4\ln(n/2))^{-1}$.
	\label{thm:upperBoundOfZ}
\end{theorem}
\begin{proof}
	Let $P_{ui}=\{v\in V|d_{uv}=i\}$ for any $u\in V$ and all $1\le i\le n$.
	Given the Definition~\ref{dfn:normalizingFactor}, \[Z=\left(\sum\limits_{v\in V\setminus\{u\}}d_{uv}^{-2}\right)^{-1}=\left(\sum\limits_{i=1}^n|P_{ui}|\cdot i^{-2}\right)^{-1}<\left(\sum\limits_{i=1}^{\lceil n/2\rceil-1}|P_{ui}|\cdot i^{-2}\right)^{-1}=(4H_{\lceil n/2\rceil-1})^{-1},\]
	due to $i\le\lceil n/2\rceil-1<n/2$ and Lemma~\ref{lem:numberOfVerticesAtDistanceI}.
	By Fact~\ref{fct:harmonicNumber} and $\lceil n/2\rceil\ge n/2$, $Z<(4\ln(n/2))^{-1}$.
\end{proof}

\begin{theorem}
	The normalizing factor $Z>(4(\ln n+1))^{-1}$.
	\label{thm:lowerBoundOfZ}
\end{theorem}
\begin{proof}
	As in the beginning of the Theorem~\ref{thm:upperBoundOfZ} proof and by Lemma~\ref{lem:numberOfVerticesAtDistanceI}, $Z>\left(\sum\limits_{i=1}^n4i\cdot i^{-2}\right)^{-1}=\left(4H_n\right)^{-1}$.
	By Fact~\ref{fct:harmonicNumber}, $Z>(4(\ln n+1))^{-1}$.
\end{proof}

Despite that Theorem~\ref{thm:lowerBoundOfZ} is not used in the rest of the paper, it implies that the upper bound of $Z$, presented in Theorem~\ref{thm:upperBoundOfZ}, is tight.
\section{Toroidal Small World Labeling Problem}
\label{sec:toroidalSmallWorldLabelingProblem}

A \emph{position}, or \emph{label}, $p=(p_1,p_2)$ in a two-dimensional torus of size $n$ is an element of $[\![n]\!]^2$, and $p_i$, for $i\in\{1,2\}$, are the \emph{coordinates} of $p$.

\begin{definition}
	Let $d'_n:[\![n]\!]^2\times[\![n]\!]^2\rightarrow\naturals$ be the function such that $d'_n(x,y)=\sum_{i=1}^2\min(|x_i-y_i|,n-|x_i-y_i|)$, where $x,y\in[\![n]\!]^2$.
	A \emph{two-dimensional toroidal vertex labeling function} of $T$ is a function $\ell_T:V\rightarrow[\![n]\!]^2$ such that $d_{uv}=d'_n(\ell_T(u),\ell_T(v))$ for all $u,v\in V$ and $n=|V|^{1/2}$.
	\label{dfn:2DimensionalToroidalVertexLabelingFunction}
\end{definition}

The function $d'$ is similar to $d$, but $d$ defines the distance between a pair of vertices on $T$ and $d'$ defines the distance between a pair of positions (labels).
The function $\ell_T$ defines a label to each vertex of $T$ such that the distances $d$ and $d'$ between any two vertices of $T$ are equal.
The \emph{toroidal small world labeling problem} consists of finding a two-dimensional toroidal vertex labeling function $\ell_{T'}$ for a \UTSW graph $G$ and a spanning torus $T'$ of size $n$.
Note that $T'$ is a spanning torus of $G$ that may be distinct from the original torus $T$ generated by \UTSW model.
This is an interesting problem because, given a two-dimensional toroidal vertex labeling function $\ell_{T'}$, it is possible to run myopic search for routing messages in $G$.

Note that $G$ has at least one two-dimensional spanning torus, due to the \UTSW definition.
Besides that, the long-range edges creation process in \UTSW may generate others two-dimensional spanning tori, as we explain next.
The original torus $T$ is composed only by local edges, while the others tori are composed by long-range edges also.
The left side of the Figure~\ref{fig:Ambiguity} shows a part of a \UTSW graph.
In this case, each vertex chose a specific vertex to create a long-range edge, represented by the dashed arrows.
This process generates a graph that may have more than one spanning torus.
The Figure~\ref{fig:Ambiguity} highlights one of them in the right side.

\begin{figure}
	\centering
	\includegraphics{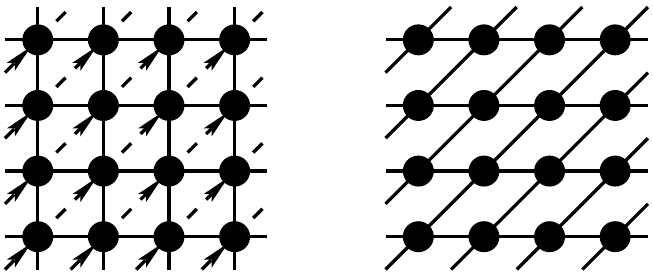}
	\caption{Graph with more than one spanning torus.}
	\label{fig:Ambiguity}
\end{figure}
\section{Cycles of Size Four Outside the Torus}
\label{sec:cyclesOfSize4OutsideTheTorus}

The identification of a spanning two-dimensional torus is an important procedure.
It can be used by a breadth-first search to label the vertices, as we do in Section~\ref{sec:labelingAlgorithm}.
Our approach is to detect cycles of size four composed by distinct vertices, due to a two-dimensional torus being an arrange of them.

A \UTSW graph may have induced four-cycles composed by some long-range edges, instead of only local edges.
The existence of this type of cycles is a problem in the torus identification process, if one uses the approach of searching for induced four-cycles.
For any $u\in V$, there are only four cycles of size four rooted in $u$ in the torus $T$.
Besides that, the \UTSW process may create long-range edges such that more induced four-cycles rooted in $u$ appear in $G$.
Despite this, we show in Theorem~\ref{thm:upperBoundOfEu} that the probability of the existence of four-cycles composed by at least one long-range edge is small, and this is the main result of this section.

We aim to bound the probability of the event that at least one four-cycle, rooted in $u$, is in $G$ but not is in $T$.
Such event, turns out to be a union of nine other events.
Lemmas~\ref{lem:E1uProbability} to \ref{lem:E9uProbability} show upper bounds on the probabilities of each of these nine events.
Lemmas~\ref{lem:sumOfDistances1} and \ref{lem:sumOfDistances2} are technical results and, together with the Fact~\ref{fct:RiemannZetaFunctionOf3}, are used to prove the other results.

\begin{lemma}
	$\sum\limits_{v\in V\setminus\{u\}}d_{uv}^{-2}<4(\ln n+1)$ for all $u\in V$.
	\label{lem:sumOfDistances1}
\end{lemma}
\begin{proof}
	Grouping the terms of the sum by their values and using Lemma~\ref{lem:numberOfVerticesAtDistanceI}, $\sum\limits_{v\in V\setminus\{u\}}d_{uv}^{-2}<\sum\limits_{i=1}^n4i\cdot i^{-2}$.
	By Fact~\ref{fct:harmonicNumber}, $\sum\limits_{v\in V\setminus\{u\}}d_{uv}^{-2}<4(\ln n+1)$.
\end{proof}

\begin{fact}[Riemann zeta function]
	$\zeta(3)=\sum\limits_{i=1}^{\infty}i^{-3}<1.20206,$ where $\zeta$ is the \emph{Riemann zeta function}.
	\label{fct:RiemannZetaFunctionOf3}
\end{fact}

\begin{lemma}
	$\sum\limits_{w\in V\setminus\{u,v\}}d_{uw}^{-2}d_{wv}^{-2}<12\zeta(3)+4$ for all $u,v\in V$ and $u\ne v$.
	\label{lem:sumOfDistances2}
\end{lemma}
\begin{proof}
	The distance $d_{uv}\le n$ because $T$ is a two-dimensional torus.
	Thus, the sum can be split into three sums:
	(i) $\sum\limits_{i=1}^{d_{uv}-1}\sum\limits_{\substack{w\in V\setminus\{u,v\}\\d_{uw}=i}}d_{uw}^{-2}d_{wv}^{-2}$, (ii) $\sum\limits_{\substack{w\in V\setminus\{u,v\}\\d_{uw}=d_{uv}}}d_{uw}^{-2}d_{wv}^{-2}$ and (iii) $\sum\limits_{i=d_{uv}+1}^n\sum\limits_{\substack{w\in V\setminus\{u,v\}\\d_{uw}=i}}d_{uw}^{-2}d_{wv}^{-2}$.

	For (i), $d_{uw}<d_{uv}$ for all terms.
	This fact allows the using of the triangle inequality such that $d_{uv}-d_{uw}$ is positive for all $w\in V\setminus\{u,v\}$.
	So, this term is at most $\sum\limits_{i=1}^{d_{uv}-1}\sum\limits_{\substack{w\in V\setminus\{u,v\}\\d_{uw}=i}}d_{uw}^{-2}(d_{uv}-d_{uw})^{-2}$.
	Since $d_{uw}=i$ and by Lemma~\ref{lem:numberOfVerticesAtDistanceI}, \begin{eqnarray*}
		\sum\limits_{i=1}^{d_{uv}-1}\sum\limits_{\substack{w\in V\setminus\{u,v\}\\d_{uw}=i}}d_{uw}^{-2}d_{wv}^{-2}
		    &\le& 4\sum\limits_{i=1}^{d_{uv}-1}i^{-1}(d_{uv}-i)^{-2}\\
		    &=& 4\left(\sum\limits_{i=1}^{\left\lfloor\frac{d_{uv}-1}{2}\right\rfloor}i^{-1}(d_{uv}-i)^{-2}+\sum\limits_{i=1}^{\left\lceil\frac{d_{uv}-1}{2}\right\rceil}(d_{uv}-i)^{-1} i^{-2}\right)\\
		    &\le& 4\left(\sum\limits_{i=1}^{\left\lfloor\frac{d_{uv}-1}{2}\right\rfloor}i^{-1} i^{-2}+\sum\limits_{i=1}^{\left\lceil\frac{d_{uv}-1}{2}\right\rceil}i^{-1} i^{-2}\right),
	\end{eqnarray*}
	where the last inequality holds because $d_{uv}-i\ge i$ for all terms in both sums.
	As all sums terms are positive for $i\ge1$, then $\sum\limits_{i=1}^{d_{uv}-1}\sum\limits_{\substack{w\in V\setminus\{u,v\}\\d_{uw}=i}}d_{uw}^{-2} d_{wv}^{-2}<8\sum\limits_{i=1}^{\infty}i^{-3}=8\zeta(3)$, by Fact~\ref{fct:RiemannZetaFunctionOf3}.

	For (ii), as $d_{wv}\ge1$, $d_{uw}=d_{uv}$, $d_{uv}\ge1$ and Lemma~\ref{lem:numberOfVerticesAtDistanceI},
	\[\sum\limits_{\substack{w\in V\setminus\{u,v\}\\d_{uw}=d_{uv}}}d_{uw}^{-2} d_{wv}^{-2}\le\sum\limits_{\substack{w\in V\setminus\{u,v\}\\d_{uw}=d_{uv}}}d_{uw}^{-2}\le4d_{uv}\cdot d_{uv}^{-2}\le4.\]

	For (iii), $\sum\limits_{i=d_{uv}+1}^n\sum\limits_{\substack{w\in V\setminus\{u,v\}\\d_{uw}=i}}d_{uw}^{-2} d_{wv}^{-2}\le4\sum\limits_{i=d_{uv}+1}^ni^{-1}(i-d_{uv})^{-2}$ holds by a similar way of (i).
	The latter is equal to $4\sum\limits_{i=1}^{n-d_{uv}}(d_{uv}+i)^{-1} i^{-2} < 4\sum\limits_{i=1}^{n-d_{uv}}i^{-1} i^{-2}<4\sum\limits_{i=1}^{\infty}i^{-3}=4\zeta(3)$, by Fact~\ref{fct:RiemannZetaFunctionOf3}.
	Given the upper bounds on (i), (ii) and (iii), therefore $\sum\limits_{w\in V\setminus\{u,v\}}d_{uw}^{-2} d_{wv}^{-2}<12\zeta(3)+4$.
\end{proof}

Given $u\in V$, let $E_u$ be the event of the existence of at least one four-cycle, rooted in $u$, that is in $G$ but not in $T$.
Note that this event is equal to the event of \UTSW creating long-range edges so that exists at least one four-cycle rooted in $u$ that belongs to $G$ composed by \emph{at least one long-range edge}.
Let $s$ be a local edge and $w$ be a long-range edge in a sequence of edges of a four-cycle rooted in $u$.
Note that, there are $2^4$ possible combinations of local and long-range edges and, one of them is certainly not in $E_u$, which is $(s,s,s,s)$, \textit{i.e.} the four-cycle with only local edges.
The others $15$ combinations can be grouped in nine sub-events of $E_u$, where each group leads to a different analysis, but the events belonging to a given group have the same analysis (which motivates the grouping).
Then, $E_u=\bigcup_{i=1}^9E_{iu}$ and, each $E_{iu}$ is defined by the following list.
\begin{itemize}
	\item $E_{1u}$: existence of at least one four-cycle, rooted in $u$, that belongs to $G$ with edge sequence $(s,s,s,w)$ or $(w,s,s,s)$;
	\item $E_{2u}$: existence of at least one four-cycle, rooted in $u$, that belongs to $G$ with edge sequence $(s,s,w,s)$ or $(s,w,s,s)$;
	\item $E_{3u}$: existence of at least one four-cycle, rooted in $u$, that belongs to $G$ with edge sequence $(s,s,w,w)$ or $(w,w,s,s)$;
	\item $E_{4u}$: existence of at least one four-cycle, rooted in $u$, that belongs to $G$ with edge sequence $(s,w,s,w)$ or $(w,s,w,s)$;
	\item $E_{5u}$: existence of at least one four-cycle, rooted in $u$, that belongs to $G$ with edge sequence $(s,w,w,s)$;
	\item $E_{6u}$: existence of at least one four-cycle, rooted in $u$, that belongs to $G$ with edge sequence $(s,w,w,w)$ or $(w,w,w,s)$;
	\item $E_{7u}$: existence of at least one four-cycle, rooted in $u$, that belongs to $G$ with edge sequence $(w,s,s,w)$;
	\item $E_{8u}$: existence of at least one four-cycle, rooted in $u$, that belongs to $G$ with edge sequence $(w,s,w,w)$ or $(w,w,s,w)$;
	\item $E_{9u}$: existence of at least one four-cycle, rooted in $u$, that belongs to $G$ with edge sequence $(w,w,w,w)$.
\end{itemize}

Lemmas~\ref{lem:E1uProbability} to \ref{lem:E9uProbability} show upper bounds on the probabilities of each of these events.
Finally, Theorem~\ref{thm:upperBoundOfEu} bounds the probability of $E_u$.
Recall that $C_{uv}$ is the event of vertex $u\in V$ choosing the vertex $v\in V\setminus\{u\}$ to create a long-range edge.

\begin{lemma}
	$\Pr(E_{1u})<2/3\ln^{-1}(n/2)$.
	\label{lem:E1uProbability}
\end{lemma}
\begin{proof}
	Let $A=\{a\in V|d_{ua}=3\}$.
	Then, $\Pr(E_{1u})=\Pr\left(\bigcup\limits_{a\in A}(C_{ua}\cup C_{au})\right)$.
	Using union bound, Definition~\ref{dfn:UTSWModel}, $d_{ua}=d_{au}=3$ and $|A|\le4 d_{ua}=12$, by Lemma~\ref{lem:numberOfVerticesAtDistanceI}, then $\Pr(E_{1u})\le8/3Z$. By Theorem~\ref{thm:upperBoundOfZ}, $\Pr(E_{1u})<2/3\ln^{-1}(n/2)$.
\end{proof}

\begin{lemma}
	$\Pr(E_{2u})<64/9\ln^{-1}(n/2)$.
	\label{lem:E2uProbability}
\end{lemma}
\begin{proof}
	Let $A=\{a\in V|d_{ua}=1\}$ and $B=\{b\in V|d_{ub}=2\}$.
	So $\Pr(E_{2u})=\Pr\left(\bigcup\limits_{a\in A}\bigcup\limits_{b\in B}(C_{ab}\cup C_{ba})\right)$.
	Using union bound, Definition~\ref{dfn:UTSWModel} and  $d_{ab}=d_{ba}$, then $\Pr(E_{2u}) \le 2Z\sum\limits_{a\in A}\sum\limits_{b\in B}d_{ab}^{-2}$.
	By Lemma~\ref{lem:numberOfVerticesAtDistanceI}, $|B|\le8$.
	When $|B|=8$, $d_{ab}=3$ occurs five times and $d_{ab}=1$ occurs three times in the last sum.
	So $\Pr(E_{2u})\le256/9Z$ because $|A|=4$.
	By Theorem~\ref{thm:upperBoundOfZ}, $\Pr(E_{2u})<64/9\ln^{-1}(n/2)$.
\end{proof}

\begin{lemma}
	$\Pr(E_{3u})<(6\zeta(3)+3/4)\ln^{-2}(n/2)$.
	\label{lem:E3uProbability}
\end{lemma}
\begin{proof}
	Let $A=\{a\in V|d_{ua}=2\}$.
	The event $E_{3u}$ is
	\[\bigcup\limits_{a\in A}\bigcup\limits_{b\in V\setminus\{u,a\}}(C_{ab}\cap C_{bu})\cup(C_{ab}\cap C_{ub})\cup(C_{ba}\cap C_{ub}).\]
	Note that the three intersections are between mutually independent events.
	So, in a similar way of Lemmas~\ref{lem:E1uProbability} and \ref{lem:E2uProbability} proofs, we have \[\Pr(E_{3u})\le3Z^2\sum\limits_{a\in A}\sum\limits_{b\in V\setminus\{u,a\}}d_{ab}^{-2}d_{ub}^{-2}.\]
	The last sum can be split into two other sums, one for $d_{ub}=2$ and the other for $d_{ub}\ge3$, because there is no $b\in V\setminus\{u,a\}$ such that $d_{ub}=1$ in $E_{3u}$.
	When $d_{ub}=2$, $d_{ub}\le d_{ab}$ for all $a\in A$ and, as a consequence, $\sum\limits_{b\in V\setminus\{u,a\}}d_{ab}^{-2}d_{ub}^{-2}\le4\cdot2\cdot 2^{-2}\cdot 2^{-2}=1/2$, due to Lemma~\ref{lem:numberOfVerticesAtDistanceI}.
	When $d_{ub}\ge3$, the triangle inequality is used to find $d_{ab}\ge d_{ub}-2$, because $d_{ua}=2$. By this fact and Lemma~\ref{lem:numberOfVerticesAtDistanceI}, $\sum\limits_{i=3}^n\sum\limits_{\substack{b\in V\setminus\{u,a\}\\d_{ub}=i}}d_{ab}^{-2} d_{ub}^{-2}\le\sum\limits_{i=3}^n4i\cdot(i-2)^{-2}\cdot i^{-2}<4\sum\limits_{i=1}^\infty i^{-3}$.
	Therefore, $\Pr(E_{3u})<3Z^2\cdot8\left(1/2+4\sum\limits_{i=1}^\infty i^{-3}\right)$, by the fact that $|A|\le4d_{ua}=8$, due to Lemma~\ref{lem:numberOfVerticesAtDistanceI}.
	Using Fact~\ref{fct:RiemannZetaFunctionOf3} and Theorem~\ref{thm:upperBoundOfZ}, $\Pr(E_{3u})<(6\zeta(3)+3/4)\ln^{-2}(n/2)$.
\end{proof}

\begin{lemma}
	$\Pr(E_{4u})<(16\zeta(3)+5/4)\ln^{-2}(n/2)$.
	\label{lem:E4uProbability}
\end{lemma}
\begin{proof}
	Let $A=\{a\in V|d_{ua}=1\}$ and, for a given $b\in V$, $C_b=\{c\in V|d_{bc}=1\}$.
	Note that, for $x,y\in V$, $C_{xy}$ is an event and $C_b$ is a set.
	The event $E_{4u}$ is
	\[\bigcup\limits_{a\in A}\bigcup\limits_{b\in V\setminus\{u,a\}}\bigcup\limits_{c\in C_b\setminus\{u,a\}}(C_{ab}\cup C_{ba})\cap(C_{uc}\cup C_{cu}).\]
	All simple events are mutually independent.
	Using the union bound, the simple events independence and Definition~\ref{dfn:UTSWModel}, then
	\[\Pr(E_{4u})\le4Z^2\sum\limits_{a\in A}\sum\limits_{b\in V\setminus\{u,a\}}\sum\limits_{c\in C_b\setminus\{u,a\}}d_{ab}^{-2} d_{uc}^{-2}.\]
	The second sum can be split for $d_{ab}=2$ and $d_{ab}\ge3$, because there is no $b\in V\setminus\{u,a\}$ such that $d_{ab}=1$ for all $a\in A$ in $E_{4u}$.

	When $d_{ab}=2$, $|\{b\in V\setminus\{u,a\}:d_{ab}=2\}|\le4d_{ab}=8$ holds, due to Lemma~\ref{lem:numberOfVerticesAtDistanceI}, for all $a\in A$.
	Note that, at most five of the at most eight vertices $b\in V\setminus\{u,a\}$ have their $C_b$ with size $|\{c\in C_b\setminus\{u,a\}:d_{ab}=2\}|\le4d_{bc}=4$.
	The others at most three vertices do not belong to the event $E_{4u}$.
	Furthermore, $d_{ab}\le d_{uc}$ when $d_{ab}=2$ and, so $\sum\limits_{b\in V\setminus\{u,a\}}\sum\limits_{c\in C_b\setminus\{u,a\}}d_{ab}^{-2} d_{uc}^{-2}\le5\cdot4\cdot2^{-2}\cdot2^{-2}=5/4$.

	When $d_{ab}\ge3$, $d_{uc}\ge d_{ab}-2$ for all $a\in A$.
	Combining these splittings of the sum for $d_{ab}=2$ and $d_{ab}\ge3$, together with Lemma~\ref{lem:numberOfVerticesAtDistanceI},
	\[\Pr(E_{4u})\le4Z^2\cdot4\left(5/4+16\sum\limits_{i=3}^ni^{-1}\cdot(i-2)^{-2}\right).\]
	As $\sum\limits_{i=3}^ni^{-1}\cdot(i-2)^{-2}<\sum\limits_{i=1}^\infty i^{-3}$, by Fact~\ref{fct:RiemannZetaFunctionOf3} and Theorem~\ref{thm:upperBoundOfZ}, $\Pr(E_{4u})<(16\zeta(3)+5/4)\ln^{-2}(n/2)$.
\end{proof}

\begin{lemma}
	$\Pr(E_{5u})<(9/2\zeta(3)+63/128)\ln^{-2}(n/2)$.
	\label{lem:E5uProbability}
\end{lemma}
\begin{proof}
	Let $A=\{a\in V|d_{ua}=1\}$ such that $A=\{a_1,a_2,a_3,a_4\}$ in any sequence.
	The event $E_{5u}$ is
    \[\bigcup\limits_{\substack{a_i,a_j\in A\\i<j}}\bigcup\limits_{b\in V\setminus\{u,a_i,a_j\}}(C_{a_ib}\cap C_{ba_j})\cup(C_{a_ib}\cap C_{a_jb})\cup(C_{ba_i}\cap C_{a_jb}).\]
	All simple events in the intersections are independent.
	So,
	\[\Pr(E_{5u})\le3Z^2\sum\limits_{\substack{a_i,a_j\in A\\i<j}}\sum\limits_{b\in V\setminus\{u,a_i,a_j\}}d_{a_ib}^{-2} d_{a_jb}^{-2}.\]
	Note that $d_{a_ib}\ge2$ in $E_{5u}$. When $d_{a_ib}=2$, there is a $b\in V$ such that $b=a_j$ for all $a_i,a_j\in A$.
	Thus, $|\{b\in V\setminus\{u,a_i,a_j\}:d_{a_ib}=2\}|\le4d_{a_ib}-1=7$, due to Lemma~\ref{lem:numberOfVerticesAtDistanceI}.
	Furthermore, $d_{a_ib}\le d_{a_jb}$ for all $a_i,a_j\in A$.
	Thus, the inequality $\sum\limits_{b\in V\setminus\{u,a_i,a_j\}}d_{a_ib}^{-2}d_{a_jb}^{-2}\le7\cdot2^{-2}\cdot2^{-2}=7/16$ holds when $d_{a_ib}=2$.
	When $d_{a_ib}\ge3$, $d_{a_jb}\ge d_{a_ib}-2$ for all $a_i,a_j\in A$ and $b\in V\setminus\{u,a_i,a_j\}$.
	Combining with Lemma~\ref{lem:numberOfVerticesAtDistanceI}, the inequality $\sum\limits_{b\in V\setminus\{u,a_i,a_j\}}d_{a_ib}^{-2}d_{a_jb}^{-2}\le\sum\limits_{k=3}^n4k\cdot k^{-2}\cdot(k-2)^{-2}$ holds.
	As $\sum\limits_{k=3}^nk^{-1}\cdot(k-2)^{-2}<\sum\limits_{k=1}^\infty k^{-3}$, then $\Pr(E_{5u})<3Z^2\sum\limits_{\substack{a_i,a_j\in A\\i<j}}\left(7/16+4\sum\limits_{k=1}^\infty k^{-3}\right)$.
	By Lemma~\ref{lem:numberOfVerticesAtDistanceI}, $|A|=4d_{ua}=4$, thus there are $\binom{4}{2}$ distinct combinations of values for $i$ and $j$.
	Therefore, $\Pr(E_{5u})<(9/2\zeta(3)+63/128)\ln^{-2}(n/2)$, due to Fact~\ref{fct:RiemannZetaFunctionOf3} and Theorem~\ref{thm:upperBoundOfZ}.
\end{proof}

\begin{lemma}
	$\Pr(E_{6u})<(12\zeta(3)+4)(\ln n+1)\ln^{-3}(n/2)$.
	\label{lem:E6uProbability}
\end{lemma}
\begin{proof}
	Let $A=\{a\in V|d_{ua}=1\}$.
	The event $E_{6u}$ is
	\[\bigcup_{a\in A}\bigcup_{b\in V\setminus\{u,a\}}\bigcup_{c\in V\setminus\{u,a,b\}}\substack{(C_{ab}\cap C_{bc}\cap C_{uc})\cup(C_{ab}\cap C_{cb}\cap C_{uc})\cup\\(C_{ba}\cap C_{cb}\cap C_{uc})\cup(C_{ab}\cap C_{bc}\cap C_{cu}).}\]
	All the intersections in $E_{6u}$ are among mutually independent simple events.
	So, because $d_{ab}=d_{ba}$, $d_{bc}=d_{cb}$, $d_{uc}=d_{cu}$ and by Definition~\ref{dfn:UTSWModel}, the probabilities of the four combinations of intersections of simple events in $E_{6u}$ are equal to $Z^3\cdot d_{ab}^{-2}\cdot d_{bc}^{-2}\cdot d_{cu}^{-2}$, for all $a\in A$, $b\in V\setminus\{u,a\}$ and $c\in V\setminus\{u,a,b\}$.
	Using union bound,
	\[\Pr(E_{6u})\le4Z^3\sum\limits_{a\in A}\sum\limits_{b\in V\setminus\{u,a\}}d_{ab}^{-2}\sum\limits_{c\in V\setminus\{u,a,b\}}d_{bc}^{-2} d_{cu}^{-2}.\]
	Using Lemmas~\ref{lem:sumOfDistances2} and \ref{lem:sumOfDistances1} and due to $|A|=4d_{ua}=4$, by Lemma~\ref{lem:numberOfVerticesAtDistanceI} and Theorem~\ref{thm:upperBoundOfZ}, $\Pr(E_{6u})<(12\zeta(3)+4)(\ln n+1)\ln^{-3}(n/2)$.
\end{proof}

\begin{lemma}
	$\Pr(E_{7u})<(6\zeta(3)+21/32)\ln^{-2}(n/2)$.
	\label{lem:E7uProbability}
\end{lemma}
\begin{proof}
	For a given $a\in V$, let $B_a=\{b\in V|d_{ab}=2\}$.
	The event $E_{7u}$ is
	\[\bigcup\limits_{a\in V\setminus\{u\}}\bigcup\limits_{b\in B_a\setminus\{u\}}(C_{ua}\cap C_{bu})\cup(C_{au}\cap C_{ub})\cup(C_{au}\cap C_{bu}).\]
	All intersections are between mutually independent simple events.
	Using union bound, the events independence, Definition~\ref{dfn:UTSWModel}, $d_{au}=d_{ua}$ and $d_{bu}=d_{ub}$, then
	\[\Pr(E_{7u})\le3Z^2\sum\limits_{a\in V\setminus\{u\}}\sum\limits_{b\in B_a\setminus\{u\}}d_{ua}^{-2} d_{ub}^{-2}.\]
	Note that $d_{ua}\ge2$ in $E_{7u}$. So, the first sum can be split for $d_{ua}=2$ and $d_{ua}\ge3$.

	When $d_{ua}=2$, there is a $b\in B_a$ such that $b=u$ for all $a\in V\setminus\{u\}$.
	As a consequence, the inequalities $|\{b\in B_a\setminus\{u\}:d_{ua}=2\}|\le4d_{ab}-1=7$ and $|\{a\in V\setminus\{u\}:d_{ua}=2\}|\le4d_{ua}=8$ hold, due to Lemma~\ref{lem:numberOfVerticesAtDistanceI}.
	Besides that, $d_{ua}\le d_{ub}$ for all $a\in V\setminus\{u\}$ and $b\in B_a\setminus\{u\}$.
	Thus, $\sum\limits_{a\in V\setminus\{u\}}\sum\limits_{b\in B_a\setminus\{u\}}d_{ua}^{-2} d_{ub}^{-2}\le8\cdot7\cdot2^{-2}\cdot2^{-2}=7/2$ when $d_{ua}=2$.

	When $d_{ua}\ge3$, then $d_{ub}\ge d_{ua}-2$ for all $a\in V\setminus\{u\}$ and $b\in B_a\setminus\{u\}$.
	Therefore,
	\[\sum\limits_{a\in V\setminus\{u\}}\sum\limits_{b\in B_a\setminus\{u\}}d_{ua}^{-2} d_{ub}^{-2}\le32\sum\limits_{i=3}^ni^{-1}\cdot(i-2)^{-2},\]
	due to the facts $|\{a\in V\setminus\{u\}:d_{ua}\ge3\}|\le4d_{ua}$ and $|\{b\in B_a\setminus\{u\}:d_{ua}\ge3\}|\le4d_{ab}=8$, by Lemma~\ref{lem:numberOfVerticesAtDistanceI}.
	Combining the cases for $d_{ua}=2$ and $d_{ua}\ge3$, then $\Pr(E_{7u})\le3Z^2\left(7/2+32\sum\limits_{i=3}^ni^{-1}\cdot(i-2)^{-2}\right)$.
	As $\sum\limits_{i=3}^ni^{-1}\cdot(i-2)^{-2}<\sum\limits_{i=1}^\infty i^{-3}$ and by Fact~\ref{fct:RiemannZetaFunctionOf3} and Theorem~\ref{thm:upperBoundOfZ}, so $\Pr(E_{7u})<(6\zeta(3)+21/32)\ln^{-2}(n/2)$.
\end{proof}

\begin{lemma}
	$\Pr(E_{8u})<(12\zeta(3)+4)(\ln n+1)\ln^{-3}(n/2)$.
	\label{lem:E8uProbability}
\end{lemma}
\begin{proof}
	For a given $a\in V$, let $B_a=\{b\in V|d_{ab}=1\}$.
	The event $E_{8u}$ is
	\[\bigcup\limits_{a\in V\setminus\{u\}}\bigcup\limits_{b\in B_a\setminus\{u\}}\bigcup\limits_{c\in V\setminus\{u,a,b\}}\substack{(C_{ua}\cap C_{bc}\cap C_{cu})\cup(C_{au}\cap C_{bc}\cap C_{uc})\cup\\(C_{au}\cap C_{bc}\cap C_{cu})\cup(C_{au}\cap C_{cb}\cap C_{uc}).}\]
	The proof follows similarly to the proof of Lemma~\ref{lem:E6uProbability}.
\end{proof}

\begin{lemma}
	$\Pr(E_{9u})<(3/4\zeta(3)+1/4)(\ln n+1)^2\ln^{-4}(n/2)$.
	\label{lem:E9uProbability}
\end{lemma}
\begin{proof}
	The event $E_{9u}$ is
	\[\bigcup\limits_{a\in V\setminus\{u\}}\bigcup\limits_{b\in V\setminus\{u,a\}}\bigcup\limits_{c\in V\setminus\{u,a,b\}}(C_{ua}\cap C_{ab}\cap C_{bc}\cap C_{cu}).\]
	All simple events in the intersections are mutually independent.
	Using union bound, the independence of the simple events and Definition~\ref{dfn:UTSWModel},
	\[\Pr(E_{9u})\le Z^4\sum\limits_{a\in V\setminus\{u\}}d_{ua}^{-2}\sum\limits_{b\in V\setminus\{u,a\}}d_{ab}^{-2}\sum\limits_{c\in V\setminus\{u,a,b\}}d_{bc}^{-2} d_{cu}^{-2}.\]
	By Lemmas~\ref{lem:sumOfDistances2} and \ref{lem:sumOfDistances1} and Theorem~\ref{thm:upperBoundOfZ}, $\Pr(E_{9u})<(3/4\zeta(3)+1/4)(\ln n+1)^2\ln^{-4}(n/2)$.
\end{proof}

\begin{theorem}
	$\Pr(E_u)=\textrm{\emph{O}}\left(\log^{-1}n\right)$.
	\label{thm:upperBoundOfEu}
\end{theorem}
\begin{proof}
	The bound for the probability of the event $E_u=\bigcup_{i=1}^9E_{iu}$ follows by union bound and Lemmas~\ref{lem:E1uProbability} to \ref{lem:E9uProbability}.
	Then,
	\begin{align*}
		\Pr(E_u)\le&\sum_{i=1}^9\Pr(E_{iu})\\
				<&70/9\ln^{-1}(n/2)+\\
				&(65/2\zeta(3)+403/128)\ln^{-2}(n/2)+\\
				&(24\zeta(3)+8)(\ln n+1)\ln^{-3}(n/2)+\\
				&(3/4\zeta(3)+1/4)(\ln n+1)^2\ln^{-4}(n/2).
	\end{align*}
	As $(\ln n+1)/\ln(n/2)$ is O$(1)$, so $\Pr(E_u)$ is O$\left(\log^{-1}n\right)$.
\end{proof}

Theorem~\ref{thm:upperBoundOfEu} means that the probability of the existence of at least one four-cycle, rooted in a vertex $u\in V$, that is in $G$ but not in $T$ is small.
Particularly, $\Pr(E_u)=0$ when $n$ goes to infinity.
Section~\ref{sec:labelingAlgorithm} presents the labeling algorithm, that uses the results of this section.
\section{Labeling Algorithm}
\label{sec:labelingAlgorithm}

The design of the labeling algorithm is based on the spanning torus identification.
We next describe its idea.
Note that a breadth-first search can be used to label the vertices of the spanning torus.
But, a \UTSW graph has a spanning torus together with long-range edges.
So, at first, the algorithm tries to remove the long-range edges.
This procedure consists in running searches rooted in each vertex aiming to remove incident long-range edges.
The algorithm runs a search with depth four rooted in each vertex, and creates a list of four-cycles rooted in it.
It identifies lattice patterns, defined in Definition~\ref{dfn:latticePattern}, through combinations of the four-cycles in the list.
If the algorithm identifies lattice patterns with only local edges incident to the root vertex, then the remaining incident edges are long-range edges and it removes them.
The resulting graph is an ``almost'' torus with possibly a few sparsely distributed long-range edges that cannot be identified by the algorithm.

After that, we choose a random vertex of the resulting graph that can be the initial reference to label the others vertices.
Finally, the algorithm labels almost all vertices running a breadth-first search in the resulting graph rooted in that chosen vertex.
The \emph{labeling algorithm} (\LA), defined in Algorithm~\ref{alg:labelingAlgorithm}, runs five important procedures:
\begin{enumerate}
	\item \emph{Four-cycles search} (\CS), defined in Algorithm~\ref{alg:4CyclesSearch};
	\item \emph{Lattice pattern recognition algorithm} (\LPRA), defined in Algorithm~\ref{alg:latticePatternRecognitionAlgorithm};
	\item \emph{Long-range edges removing algorithm} (\LRERA), defined in Algorithm~\ref{alg:longRangeEdgesRemovingAlgorithm};
	\item \emph{Reference system labeling algorithm} (\RSLA), defined in Algorithm~\ref{alg:referenceSystemLabelingAlgorithm};
	\item \emph{Arbitrary cross labeling algorithm} (\ACLA), defined in Algorithm~\ref{alg:arbitraryCrossLabelingAlgorithm}.
\end{enumerate}
Sections~\ref{subsec:boundedSearch} to \ref{subsec:mainAlgorithm} present the definition and analysis of each of these procedures.

\subsection{Bounded search}
\label{subsec:boundedSearch}

The \emph{four-cycles search}, \CS for short, is an algorithm that takes as input a \UTSW graph $G=(V,E)$ and a vertex $u\in V$, and outputs the set of all four-cycles in $G$, rooted in $u$, composed by four distinct vertices.
It is a depth-first search with depth $d$ limited to four, starting from $d=1$.
If a vertex $v_d$, at depth $d=4$, is neighbor of $u$, then the path from root $u$ to leaf $v_d$ corresponds to a four-cycle composed by distinct vertices and, thus, the algorithm adds this cycle to the returning set $\mathcal{S}$.

Note that each cycle appears twice in $\mathcal{S}$, in both directions, due to $G$ being an undirected graph.
Because of that, the algorithm removes all duplicates in $\mathcal{S}$.
Algorithm~\ref{alg:4CyclesSearch} presents \CS, which executes two procedures: \emph{recursive four-cycles search} (\RCS) and \emph{duplication removing algorithm} (\DRA).
We denote $N_G(w)$ as the set of neighbors of  $w\in V$ in $G$.
Lemma~\ref{lem:sizeOf4CyclesSet} shows that \CS outputs a set of four-cycles with size that tends to four.

\begin{algorithm}
	\Instance{A \UTSW graph $G=(V,E)$ and a vertex $u\in V$}
	\Returns{The set of all four-cycles in $G$ rooted in $u$ and composed by four distinct vertices}
	\BlankLine
	\Alg{\CS{$G,u$}}{
		\KwRet \DRA{\RCS{$G,(u,\emph{\nil},\emph{\nil},\emph{\nil}),1$}}
	}
	\Alg{\RCS{$G,(v_1,v_2,v_3,v_4),d$}}{
		\If{$d=4$}{
			\lIf{$v_1\in N_G(v_4)$}{\KwRet $\{(v_1,v_2,v_3,v_4)\}$}
			\KwRet $\emptyset$
		}
		$\mathcal{S}\leftarrow\emptyset$\\
		\ForEach{$w\in N_G(v_d)\setminus\{v_k|1\le k\le d\}$}{
			$v_{d+1}\leftarrow w$\\
			$\mathcal{S}\leftarrow\mathcal{S}$ $\cup$ \RCS{$G,(v_1,v_2,v_3,v_4),d+1$}
		}
		\KwRet $\mathcal{S}$
	}
	\Alg{\DRA{$\mathcal{S}$}}{
		$\mathcal{R}\leftarrow\emptyset$\\
		\ForEach{$(v_1,v_2,v_3,v_4)\in\mathcal{S}$}{
			\lIf{$(v_1,v_4,v_3,v_2)\notin\mathcal{R}$}{
				$\mathcal{R}\leftarrow\mathcal{R}\cup\{(v_1,v_2,v_3,v_4)\}$
			}
		}
		\KwRet $\mathcal{R}$
	}
	\caption{Four-cycles search}
	\label{alg:4CyclesSearch}
\end{algorithm}

\begin{lemma}
	\CS outputs a set of four-cycles with expected size of at most \emph{$4+\textrm{O}\left(\log^{-1}n\right)$} for all $u\in V$.
	\label{lem:sizeOf4CyclesSet}
\end{lemma}
\begin{proof}
	Let $C$ be the set of all possible distinct four-cycles in $G$ with root vertex $u$, with at least one long-range edge and without considering directions.
	Let $X_c$ be the random variables for all $c\in C$ that are one if \CS outputs a set that contains $c$, or zero otherwise.
	Let $L$ be the random variable that counts the number of four-cycles in the set returned by \CS.
	As shown in Section~\ref{sec:cyclesOfSize4OutsideTheTorus}, the set of four-cycles with at least one long-range edge can be partitioned in nine subsets according to the sequence of edges between the vertices.
	Let $C_i$ be the subsets of $C$ partitioned as the events $E_{iu}$, for all $1\le i\le9$.
	Note that $C$ and $C_i$ are sets and $C_{uv}$ is the event of $u$ choosing $v\in V\setminus\{u\}$.
	As \CS finds at least the four distinct cycles that belongs to the torus $T$, so $L=4+\sum\limits_{c\in C}X_c$.
	Then, by the linearity of expectation,
	\[\E[L]=4+\E\left[\sum\limits_{i=1}^9\sum\limits_{c\in C_i}X_c\right]=4+\sum\limits_{i=1}^9\sum\limits_{c\in C_i}\Pr(X_c=1).\]
	We next bound $\sum_{c\in C_1}\Pr(X_c=1)$, in a similar way of Lemma~\ref{lem:E1uProbability} proof.
	Note that there are at most three distinct cycles $c\in C_1$ with the same $c_4$, where $c_k$ is the vertex in the position $k$ of the cycle $c$  and $c_1=u$.
	Let $A=\{a\in V|d_{ua}=3\}$.
	Then $\sum_{c\in C_1}\Pr(X_c=1)\le\sum_{a\in A}3\Pr(C_{ua}\cup C_{au})$.
	As $|A|\le4 d_{ua}=12$, by Lemma~\ref{lem:numberOfVerticesAtDistanceI}, $d_{ua}=d_{au}=3$ and using union bound and Definition~\ref{dfn:UTSWModel}, then $\sum_{c\in C_1}\Pr(X_c=1)\le3\cdot\Pr(E_{1u})$.
	Such proof strategy can be used to find upper bounds on $\sum_{c\in C_i}\Pr(X_c=1)$, for all $1\le i\le9$, using Lemmas~\ref{lem:E1uProbability} to \ref{lem:E9uProbability}.
	So, we have,
		\begin{multicols}{2}
			\begin{enumerate}
				\item $\sum\limits_{c\in C_1}\Pr(X_c=1)\le3\cdot\Pr(E_{1u})$;
				\item $\sum\limits_{c\in C_2}\Pr(X_c=1)\le2\cdot\Pr(E_{2u})$;
				\item $\sum\limits_{c\in C_3}\Pr(X_c=1)\le2\cdot\Pr(E_{3u})$;
				\item $\sum\limits_{c\in C_4}\Pr(X_c=1)\le\Pr(E_{4u})$;
				\item $\sum\limits_{c\in C_5}\Pr(X_c=1)\le\Pr(E_{5u})$;
				\item $\sum\limits_{c\in C_6}\Pr(X_c=1)\le\Pr(E_{6u})$;
				\item $\sum\limits_{c\in C_7}\Pr(X_c=1)\le2\cdot\Pr(E_{7u})$;
				\item $\sum\limits_{c\in C_8}\Pr(X_c=1)\le\Pr(E_{8u})$;
				\item $\sum\limits_{c\in C_9}\Pr(X_c=1)\le\Pr(E_{9u})$.
			\end{enumerate}
		\end{multicols}
	The proof follows similarly to the proof of Theorem~\ref{thm:upperBoundOfEu}.
	Then,
		\begin{align*}
			\E[L]< 4 + &146/9\ln^{-1}(n/2)+\\
					&(89/2\zeta(3)+583/128)\ln^{-2}(n/2)+\\
					&(24\zeta(3)+8)(\ln n+1)\ln^{-3}(n/2)+\\
					&(3/4\zeta(3)+1/4)(\ln n+1)^2\ln^{-4}(n/2)\textrm{, for all }u\in V.
		\end{align*}
\end{proof}

The expected size of the set returned by \CS, denoted $\E[L]$ and bounded in the end of the Lemma~\ref{lem:sizeOf4CyclesSet} proof, decreases as a reciprocal logarithmic function to the constant four whereas $n$ increases linearly.
So, $\E[L]$ assumes the greatest value when $n$ assumes the smallest value.
When $n=3$, the expected size of the set returned by \CS is $\E[L]<1745$, but better values are obtained by increasing $n$.
For example, we perform some experimental simulations, where \CS outputs sets with approximately size of $7.48$ for a \UTSW graph with $n=10$.
This value decreases for greater values of $n$, as foreseen by the theoretical bound of Lemma~\ref{lem:sizeOf4CyclesSet}.
It is approximately $5.42$ for $n=100$ and $5.24$ for $n=150$.
Next, Lemma~\ref{lem:runningTime4CyclesSearch} uses the results of the Lemma~\ref{lem:expectedDegree} and shows that \CS runs in expected constant time.

\begin{lemma}
	Let $\delta_G(u)$ be the degree of each $u\in V$ in $G$, the expected degree of $u$ is \emph{$\E[\delta_G(u)]\le6$}.
	\label{lem:expectedDegree}
\end{lemma}
\begin{proof}
	Let $N_T(u)$ be the set of neighbors of $u$ in the torus $T$.
	In \UTSW, the vertex $u$ chooses one vertex $v\in V\setminus\{u\}$ and any vertex $w\in V\setminus\{u\}$ can choose $u$ to create a long-range edge.
	Let $X_u$ be the random variable that counts the number of vertices that chose $u$ to create a long-range edge.
	Let $X_{wu}$ be the random variable that is one if the vertex $w\in V\setminus\{u\}$ choose $u$ to create a long-range edge, or zero otherwise.
	As $X_u=\sum\limits_{w\in V\setminus\{u\}}X_{wu}$, by the linearity of expectation, $\E[X_u]=\sum\limits_{w\in V\setminus\{u\}}\Pr(X_{wu}=1)$.
	Using the Definitions~\ref{dfn:UTSWModel} and \ref{dfn:normalizingFactor} and $d_{uw}=d_{wu}$, then $\E[X_u]=1$.
	Thus, using the linearity of expectation, $\E[\delta_G(u)]\le|N_T(u)|+1+\E[X_u]=6$.
\end{proof}

\begin{lemma}
	\CS runs in expected constant time in any vertex $u\in V$.
	\label{lem:runningTime4CyclesSearch}
\end{lemma}
\begin{proof}
	Let $X_d$ be the random variable that counts the number of vertices at levels $1\le d\le4$ of \RCS.
	Let $x_{d-1}$ be the already known number of vertices at level $d-1$.
	Let $D_k$ be the random variable of the degree of each vertex $1\le k\le x_{d-1}$ at level $d-1$.
	Note that $X_0=x_0=1$ and $X_d=\sum_{k=1}^{x_{d-1}}D_k$ for all $1\le d\le4$.
	Then
	\begin{eqnarray*}
		\E[X_d|X_{d-1}=x_{d-1}] &=& \E\left[\sum\limits_{k=1}^{x_{d-1}}D_k|X_{d-1}=x_{d-1}\right]\\
								&=& \sum\limits_{k=1}^{x_{d-1}}\E[D_k|X_{d-1}=x_{d-1}]\\
								&=& \sum\limits_{k=1}^{x_{d-1}}\E[D_k]\\
								&\le& 6 x_{d-1},
	\end{eqnarray*}
	due to the linearity of expectation, to the independence between the variables $D_k$ and $X_{d-1}$ for all $1\le k\le x_{d-1}$ and $1\le d\le4$, and because $\E[D_k]=\E[\delta_G(u)]\le6$ for all $u\in V$, by Lemma~\ref{lem:expectedDegree}.
	So,
	\begin{eqnarray*}
		\E[X_d] &=& \sum\limits_{x_{d-1}\ge0}\E[X_d|X_{d-1}=x_{d-1}]\Pr(X_{d-1}=x_{d-1})\\
				&\le& 6\sum\limits_{x_{d-1}\ge0}x_{d-1}\Pr(X_{d-1}=x_{d-1})\\
				&=& 6\E[X_{d-1}].
	\end{eqnarray*}
	Solving the recurrence with base $\E[X_0]=1$, we have $\E[X_d]\le6^d$ for $d\ge0$.
	Let $X$ be the random variable that counts the total number of vertices at levels $0\le d\le4$ of the \RCS.
	As $X=\sum_{d=0}^4X_d$, then $\E[X]\le\sum_{d=0}^46^d=1555$.
	The result follows because \RCS visits at most an expected constant number of vertices, as shown in the bound of $\E[X]$, and because \DRA takes as input a list of cycles with at most constant expected size, by Lemma~\ref{lem:sizeOf4CyclesSet}.
\end{proof}

\subsection{Lattice pattern and detection}

\begin{definition}
	Given a set $\mathcal{C}$ of four cycles of size four, the \emph{lattice pattern} property consists in the existence of a sequence $\left(C^0,C^1,C^2,C^3\right)$ of the four cycles $C^k\in\mathcal{C}$ where for all $0\le k\le3$:
	\begin{enumerate}[(i)]
		\item the cycle $C^k$ has a same vertex $u$;\label{dfn:latticePattern1}
		\item the pair of consecutive cycles $C^k$ and $C^{((k+1)\mod4)}$ have a distinct intersecting edge $\{u,a_k\}$;\label{dfn:latticePattern2}
		\item the cycle $C^k$ has a distinct vertex $b_k\ne u$ neighbor of $a_{((k-1)\mod4)}$ and $a_k$;\label{dfn:latticePattern3}
		\item all nine vertices $u$, $a_k$ and $b_k$ are distinct each other.\label{dfn:latticePattern4}
	\end{enumerate}
	\label{dfn:latticePattern}
\end{definition}

A sequence of cycles with the properties of the Definition~\ref{dfn:latticePattern} may start with any cycle in $\mathcal{C}$.
In fact, if $\mathcal{C}$ is a lattice pattern, then there are four distinct sequences with the lattice pattern properties.
Note that if the sequence $\left(C^0,C^1,C^2,C^3\right)$ is a lattice pattern, then are the sequences $\left(C^1,C^2,C^3,C^0\right)$, $\left(C^2,C^3,C^0,C^1\right)$ and $\left(C^3,C^0,C^1,C^2\right)$ also.
Figure~\ref{fig:LatticePattern} illustrates a lattice pattern with a sequence of cycles in the clockwise.

The \emph{lattice pattern recognition algorithm}, \LPRA for short, takes as input a set $\mathcal{C}$ of four cycles of size four, with the same root vertex $u\in V$ and composed by distinct vertices, and outputs a boolean value that informs if the cycles define a lattice pattern.
It recognizes a lattice pattern finding a sequence of the cycles in $\mathcal{C}$ that has the four properties of the Definition~\ref{dfn:latticePattern}.
The property (\ref{dfn:latticePattern1}) of the lattice pattern is assumed to be true by the input constraints.

\begin{figure}
	\centering
	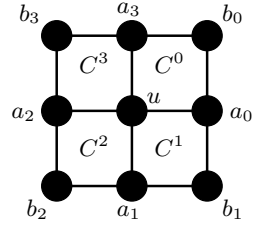
	\caption{Four four-cycles that constitutes a lattice pattern.}
	\label{fig:LatticePattern}
\end{figure}

\LPRA, defined in Algorithm~\ref{alg:latticePatternRecognitionAlgorithm}, starts selecting an arbitrary cycle $C^0\in\mathcal{C}$, which is the first cycle of the sequence $\left(C^0,C^1,C^2,C^3\right)$ without loss of generality.
Note that $C^k=\left(c^k_1,c^k_2,c^k_3,c^k_4\right)$ is the cycle $0\le k\le3$ of the sequence and $c^k_i$ is the vertex $1\le i\le4$ of the cycle $k$ of the sequence.
Moreover, $c^k_1=u$ for all $k$, by the input constraints, and either $a_0=c^0_2$ or $a_0=c^0_4$.
Without loss of generality, \LPRA sets $a_0$ to $c^0_4$, removes $C^0$ from $\mathcal{C}$ and initializes the set $S$ on lines \ref{alg:latticePatternRecognitionAlgorithm1} to \ref{alg:latticePatternRecognitionAlgorithm2}.
In this case, $a_0=c^0_4$, $b_0=c^0_3$ and $a_3=c^0_2$ in the lattice pattern shown in Figure~\ref{fig:LatticePattern}.
\LPRA uses the set $S$ to check the property (\ref{dfn:latticePattern4}) in each iteration of the line \ref{alg:latticePatternRecognitionAlgorithm3} loop.

After that, \LPRA tries to find iteratively the next cycle of sequence $C^k$ through the edge $\{u,a_{k-1}\}$, keeping the property (\ref{dfn:latticePattern2}) true.
Then, either $a_{k-1}=c^k_2$ or $a_{k-1}=c^k_4$, because $c^k_1=u$ and considering both directions of the cycle.
If there is no $C^k\in\mathcal{C}$ such that, then there is no sequence of the cycles in $\mathcal{C}$ with the properties of the Definition~\ref{dfn:latticePattern} and $\mathcal{C}$ does not define a lattice pattern.

\begin{algorithm}
	\Instance{A set $\mathcal{C}$ of four cycles of size four with the same root vertex and composed by distinct vertices}
	\Returns{A boolean value that informs if $\mathcal{C}$ defines a lattice pattern}
	\BlankLine
	\Alg{\LPRA{$\mathcal{C}$}}{
		Choose an arbitrary cycle $C^0\in\mathcal{C}$, where $C^0=\left(c^0_1,c^0_2,c^0_3,c^0_4\right)$\\
		$\mathcal{C}\leftarrow\mathcal{C}\setminus\left\{C^0\right\}$\\\label{alg:latticePatternRecognitionAlgorithm1}
		$S\leftarrow\left\{c^0_1,c^0_3,c^0_4\right\}$\\\label{alg:latticePatternRecognitionAlgorithm5}
		$a_0\leftarrow c^0_4$\\\label{alg:latticePatternRecognitionAlgorithm2}
		\For{$k\leftarrow1$  \KwTo $3$}{\label{alg:latticePatternRecognitionAlgorithm3}
			Find a cycle $C^k\in\mathcal{C}$, where $C^k=\left(c^k_1,c^k_2,c^k_3,c^k_4\right)$, such that $c^k_2=a_{k-1}$ \textbf{or} $c^k_4=a_{k-1}$\\
			\lIf{\emph{such cycle} $C^k$ \emph{does not exist}}{\KwRet \false}
			\lIf{$\left\{c^k_2,c^k_3,c^k_4\right\}\setminus\{a_{k-1}\}\cap S\ne\emptyset$}{\KwRet \false}\label{alg:latticePatternRecognitionAlgorithm4}
			$\mathcal{C}\leftarrow\mathcal{C}\setminus\left\{C^k\right\}$\\\label{alg:latticePatternRecognitionAlgorithm7}
			$S\leftarrow S\cup\left\{c^k_2,c^k_3,c^k_4\right\}$\\
			$a_k\leftarrow\left\{c^k_2,c^k_4\right\}\setminus\{a_{k-1}\}$\label{alg:latticePatternRecognitionAlgorithm8}
		}
		\lIf{$a_3=c^0_2$}{\KwRet \true}\label{alg:latticePatternRecognitionAlgorithm6}
		\KwRet \false
	}
	\caption{Lattice pattern recognition algorithm}
	\label{alg:latticePatternRecognitionAlgorithm}
\end{algorithm}

Line \ref{alg:latticePatternRecognitionAlgorithm4} checks the property (\ref{dfn:latticePattern4}) in the cycles $C^0$ to $C^k$ in each iteration $1\le k\le3$.
When $k=3$ and $\mathcal{C}$ defines a lattice pattern, either $c^3_2=c^0_2$ or $c^3_4=c^0_2$.
This is the reason that \LPRA does not add $c^0_2$ in $S$ on line \ref{alg:latticePatternRecognitionAlgorithm5}.
However, this equality is checked on line \ref{alg:latticePatternRecognitionAlgorithm6}, keeping the property (\ref{dfn:latticePattern2}) true.
Lines from \ref{alg:latticePatternRecognitionAlgorithm7} to \ref{alg:latticePatternRecognitionAlgorithm8} removes $C^k$ from $\mathcal{C}$ for \LPRA selecting $C^k$ only once; adds the vertices of $C^k$ in $S$ for checking the property (\ref{dfn:latticePattern4}) and sets $a_k$ for finding the $C^k$ of the next iteration.
In the end, \LPRA computes a sequence $\left(C^0,C^1,C^2,C^3\right)$ and checks on line \ref{alg:latticePatternRecognitionAlgorithm6} the property (\ref{dfn:latticePattern2}) for the consecutive cycles $C^3$ and $C^0$.

\LPRA does not consider the property (\ref{dfn:latticePattern3}).
As all vertices of $C^k$ are distinct each other, for all $C^k\in\mathcal{C}$, by the input constraints, so the property (\ref{dfn:latticePattern3}) is true for all $0\le k\le3$.
Therefore, as the properties (\ref{dfn:latticePattern1}) and (\ref{dfn:latticePattern3}) are assumed to be true, the cycles $C^k$ are found based on property (\ref{dfn:latticePattern2}), and the property (\ref{dfn:latticePattern4}) is checked in each iteration, then \LPRA recognizes lattice patterns.
Moreover, the running time of \LPRA is constant, due to the set $\mathcal{C}$ having size four.

\subsection{Removing long-range edges}

Note that \CS outputs a set of four-cycles with the same root vertex and composed by distinct vertices.
Also, this set has at least four cycles that define a lattice pattern, because a \UTSW graph $G$ has a two-dimensional spanning torus $T$.
Then, it is possible to check the property of lattice pattern in all combinations of four cycles in the set that \CS outputs.
The design of the next algorithm is based on this and its aim is to remove a large fraction of the long-range edges in $G$.

Let $\mathcal{L}$ be the set of four-cycles returned by \CS with input $G$ and the root vertex $u\in V$.
Let $\mathcal{C}_1,\mathcal{C}_2,\ldots,\mathcal{C}_k\subseteq\mathcal{L}$ be all the $k\ge1$ distinct combinations of four distinct cycles in $\mathcal{L}$ with lattice pattern.
Let $U_i=(V_i,E_i)$ be the graph induced by the edges of the cycles in $\mathcal{C}_i$, for all $1\le i\le k$.
Let $U=\left(\bigcup_{i=1}^kV_i\;,\;\bigcup_{i=1}^kE_i\right)$.

All vertices in the torus $T$ have degree four, because $T$ is a two-dimensional torus.
So, if $\delta_U(u)=4$, then all edges incident to $u$ in $U$ are the local edges of $u$ in $T$, recalling that $\delta_H(v)$ is the degree of $v$ in the graph $H$.
As a consequence, the set of all long-range edges of $u$ in $G$ is $\partial_G(u)\setminus\partial_U(u)$, where $\partial_{H}(v)$ is the set of all edges incident to $v$ in the graph $H$.
We refer a vertex $u\in V$ \emph{detected} if $\delta_U(u)=4$ and this is a boolean attribute of each vertex of $G$.
Note that some vertices of $G$ may not be detected.

The \emph{long-range edges removing algorithm}, \LRERA for short, takes as input a \UTSW graph $G$ and outputs a graph $(V,E')$ without the long-range edges of all detected vertices.
As explained, if a vertex $u\in V$ is detected, then all its incident edges in $G$ can be identified as local or long-range edges.
If the \LRERA detects $u\in V$, then it does not add in $E'$ all edges in $\partial_G(u)\setminus\partial_U(u)$, adding only the local edges of $u$ in $T$.
The algorithm does this through a boolean attribute \emph{add} of the edges of $G$.

\begin{algorithm}
	\Instance{A \UTSW graph $G=(V,E)$}
	\Returns{A graph $(V,E')$ without the detected long-range edges, and the \emph{detected} attribute defined for each $u\in V$}
	\BlankLine
	\Alg{\LRERA{$G$}}{
		\lForEach{$e\in E$}{$e$.add $\leftarrow$ \true}
		\ForEach{$u\in V$}{
			$\mathcal{L}\leftarrow$ \CS{$G,u$}\\\label{alg:longRangeEdgesRemovingAlgorithm3}
			$E''\leftarrow\emptyset$\\
			\ForEach{$\mathcal{C}\in\binom{\mathcal{L}}{4}$}{\label{alg:longRangeEdgesRemovingAlgorithm5}
				\If{\LPRA{$\mathcal{C}$}}{
					$E''\leftarrow E''\cup\{$edges in $\mathcal{C}$ that are incident to $u\}$
				}
			}
			\If{$|E''|=4$}{\label{alg:longRangeEdgesRemovingAlgorithm2}
				$u$.detected $\leftarrow$ \true\\
				\lForEach{$e\in\partial_G(u)\setminus E''$}{$e$.add $\leftarrow$ \false}\label{alg:longRangeEdgesRemovingAlgorithm1}
			}
			\lElse{$u$.detected $\leftarrow$ \false}\label{alg:longRangeEdgesRemovingAlgorithm4}
		}
		$E'\leftarrow\emptyset$\\
		\ForEach{$e\in E$}{
			\lIf{$e$\emph{.add}}{$E'\leftarrow E'\cup\{e\}$}
		}
		\KwRet $(V,E')$
	}
	\caption{Long-range edges removing algorithm}
	\label{alg:longRangeEdgesRemovingAlgorithm}
\end{algorithm}

The \LRERA, defined in Algorithm~\ref{alg:longRangeEdgesRemovingAlgorithm}, initializes the \emph{add} attribute of each edge with the value \true.
After that, it tries to detect each $u\in V$.
If a vertex $u$ is detected, then it sets to {\false} the \emph{add} attribute of each long-range edge of $u$.
In the end, \LRERA adds in $E'$ all edges in $G$ where \emph{add} is \true.
Furthermore, it sets to {\true} the \emph{detected} attribute of all detected vertices, otherwise sets to {\false} for the remaining vertices.
Later, the procedures \emph{reference system labeling algorithm} (Algorithm~\ref{alg:referenceSystemLabelingAlgorithm}) and \emph{arbitrary cross labeling algorithm} (Algorithm~\ref{alg:arbitraryCrossLabelingAlgorithm}) use such attribute.

\LRERA runs \CS (Algorithm~\ref{alg:4CyclesSearch}), that outputs the set $\mathcal{L}$ of all four-cycles in $G$ with root vertex $u$.
After that, it adds in $E''$ the incident edges of $u$ for all combination $\mathcal{C}$ of four distinct cycles in $\mathcal{L}$ that \LPRA recognizes as lattice pattern.
That is, $E''$ is the set of incident edges of $u$ in the graph $U$.
So, $|E''|=\delta_U(u)$ and, as a consequence, if $|E''|=4$ then all edges in $E''$ belong to the spanning torus $T$.
Then, the algorithm sets to {\false} the \emph{add} attributes of each edge in $\partial_G(u)\setminus E''$.
Lemmas~\ref{lem:runningTimeLongRangeEdgesRemovingAlgorithm}, \ref{lem:DuProbability} and \ref{lem:numberOfDetectedVertices} show, respectively, that \LRERA runs in expected linear time, the probability of \LRERA detecting a vertex is high and \LRERA detects almost all vertices.

\begin{lemma}
	\LRERA runs in expected linear time on $|V|$.
	\label{lem:runningTimeLongRangeEdgesRemovingAlgorithm}
\end{lemma}
\begin{proof}
    \CS runs in expected constant time, by Lemma~\ref{lem:runningTime4CyclesSearch}.
    The expected size of $\mathcal{L}$ is constant, due to Lemma~\ref{lem:sizeOf4CyclesSet}.
    \LPRA running time is constant.
    For line \ref{alg:longRangeEdgesRemovingAlgorithm1}, note that, by Lemma~\ref{lem:expectedDegree}, $\E[\delta_G(u)]\le6$, and $|E''|=4$ by line \ref{alg:longRangeEdgesRemovingAlgorithm2}.
    So line \ref{alg:longRangeEdgesRemovingAlgorithm1} runs in expected constant time because $\E[|\partial_G(u)\setminus E''|]\le2$.
    Then, lines~\ref{alg:longRangeEdgesRemovingAlgorithm3} to \ref{alg:longRangeEdgesRemovingAlgorithm4} runs in expected constant time.
    The other loops run in linear time on $|V|$, because $|E|\le3|V|$ by the fact that \UTSW generates the spanning torus $T=\left(V,\hat{E}\right)$ such that $\left|\hat{E}\right|=2|V|$ and at most one long-range edge per vertex. Combining these, the result follows.
\end{proof}

\begin{lemma}
	Let $D_u$ be the event of \LRERA detecting $u\in V$ and $E_u$ be the event of the existence of at least one four-cycle, rooted in $u$, that is in $G$ but not in $T$. Then, $\Pr(D_u)\ge1-4\Pr(E_u)$.
	\label{lem:DuProbability}
\end{lemma}
\begin{proof}
	The event of \LRERA does not detect $u$, that is $\overline{D_u}$, corresponds to the existence of at least one lattice pattern in $\binom{\mathcal{L}}{4}$, in line \ref{alg:longRangeEdgesRemovingAlgorithm5}, composed by at least one four-cycle with at least one long-range edge incident to $u$.
	A lattice pattern is a well defined sequence of four cycles of size four.
	Note that the probability of the existence of at least one four-cycle in $u$ with at least one long-range edge in any position of the sequence $\left(C^0,C^1,C^2,C^3\right)$ bounds the probability of $\overline{D_u}$, because the last are particular cases of the first.
	Let $A_k$ be the event $E_u$, as defined in Section~\ref{sec:cyclesOfSize4OutsideTheTorus}, on the position $k$ of the lattice pattern sequence, for all $0\le k\le3$.
	Then, the probability of \LRERA does not detect $u$ is $\Pr\left(\overline{D_u}\right)\le\Pr\left(\bigcup_{k=0}^3A_k\right)$.
	Using union bound, $\Pr\left(\overline{D_u}\right)\le4\Pr(\E_u)$ and, then, $\Pr(D_u)\ge1-4\Pr(\E_u)$.
\end{proof}

\begin{lemma}
	The expected number of vertices that \LRERA detects is \emph{$|V|\left(1-\textrm{O}\left(\log^{-1}n\right)\right)$}.
	\label{lem:numberOfDetectedVertices}
\end{lemma}
\begin{proof}
	For each $u\in V$, let $D_u$ be a random variable, such that $D_u=1$ if \LRERA detects $u$, $D_u=0$ otherwise.
	Let $D$ be the random variable that counts the number of vertices that \LRERA detects.
	So, $D=\sum_{u\in V}D_u$ and, by the linearity of expectation, $\E[D]=\sum_{u\in V}\E[D_u]=\sum_{u\in V}\Pr(D_u=1)$.
	The result follows by Lemma~\ref{lem:DuProbability} and by the proof of the Theorem~\ref{thm:upperBoundOfEu}.
	Then,
	\begin{align*}
		\E[D]\ge|V|(1-4(&70/9\ln^{-1}(n/2)+\\
						&(65/2\zeta(3)+403/128)\ln^{-2}(n/2)+\\
						&(24\zeta(3)+8)(\ln n+1)\ln^{-3}(n/2)+\\
						&(3/4\zeta(3)+1/4)(\ln n+1)^2\ln^{-4}(n/2))).
	\end{align*}
\end{proof}

Lemma~\ref{lem:numberOfDetectedVertices} means that the number of detected vertices increases as a reciprocal logarithmic function to $|V|$ whereas $n$ increases linearly.
We perform some experimental simulations that show that \LRERA detection rate is about $88\%$, $96.04\%$ and $96.35\%$ of the vertices, respectively, for $n=10, n=100$ and $n=150$.
Indeed, \LRERA detects all vertices when $n$ goes to infinity, as Corollary~\ref{cor:allVerticesAreDetectedForInfiniteN} shows.

\begin{corollary}
	The expected number of vertices that \LRERA detects tends to $|V|$ as $n$ goes to infinity.
	\label{cor:allVerticesAreDetectedForInfiniteN}
\end{corollary}
\begin{proof}
	Let $D$ be the random variable that counts the number of vertices that \LRERA detects.
	Then,
	\begin{alignat*}{2}
		\lim\limits_{n\rightarrow\infty}\textrm{E}[D]&\ge|V|(1-4\lim\limits_{n\rightarrow\infty}(&&70/9\ln^{-1}(n/2)+\\
			& &&(65/2\zeta(3)+403/128)\ln^{-2}(n/2)+\\
			& &&(24\zeta(3)+8)(\ln n+1)\ln^{-3}(n/2)+\\
			& &&(3/4\zeta(3)+1/4)(\ln n+1)^2\ln^{-4}(n/2)))\\
			&=|V|.&&
	\end{alignat*}
\end{proof}

\subsection{Labeling the reference system}

Recall that the main purpose of the \emph{labeling algorithm} (Algorithm~\ref{alg:labelingAlgorithm}) is to find a two-dimensional toroidal vertex labeling function for a spanning torus of the \UTSW graph $G$, as we explain in Section~\ref{sec:toroidalSmallWorldLabelingProblem}.
In order to label the vertices and obtain a labeling $\ell_T$, we need a \emph{reference system}.
A \emph{reference system} in $G$ is (i) an \emph{origin} that corresponds to a vertex $o\in V$ with label $(0,0)$ and (ii) the four neighbors $o_1,o_2,o_3,o_4\in V$ of $o$ in $T$ with the labels $(0,1)$, $(1,0)$, $(0,n-1)$ and $(n-1,0)$.
We refer a vertex $u\in V$ and its four neighbors in $T$ as the \emph{cross rooted in $u$}.
The labeling algorithm initializes the reference system by finding a possible origin and labeling the cross rooted in it.

\begin{algorithm}
	\Instance{A graph (``almost'' torus) $T'=(V,E')$}
	\Returns{A queue $Q$ and the labeling array $\ell$ with the reference system vertices labeled}
	\BlankLine
	\Alg{\RSLA{$T'$}}{
		\ForEach{$u\in V$}{\label{alg:referenceSystemLabelingAlgorithm4}
			$\ell[u]\leftarrow$ \nil\\
			$u$.enqueued $\leftarrow$ \false\label{alg:referenceSystemLabelingAlgorithm5}
		}
		\Repeat
		{\emph{$o$ \textbf{and} its neighbors in $T'$ are all detected}\label{alg:referenceSystemLabelingAlgorithm2}}
		{\label{alg:referenceSystemLabelingAlgorithm1}Choose $o\in V$ independently at random}
		$\ell[o]\leftarrow(0,0)$\\
		$(o_1,o_2,o_3,o_4)\leftarrow$ \RSVFA{$T',o$}\\
		$\ell[o_1]\leftarrow(0,1)$, $\ell[o_2]\leftarrow(1,0)$, $\ell[o_3]\leftarrow(0,n-1)$, $\ell[o_4]\leftarrow(n-1,0)$\\\label{alg:referenceSystemLabelingAlgorithm3}
		$Q\leftarrow\emptyset$\\
		$o$.enqueued $\leftarrow$ \true\\
		\For{$i\leftarrow1$ \KwTo $4$}{
			$Q\leftarrow Q\cup\{o_i\}$\\
			$o_i$.enqueued $\leftarrow$ \true
		}
		\KwRet $(Q,\ell)$
	}
	\Alg{\RSVFA{$T',o$}}{
		$\mathcal{C}\leftarrow$ \CS{$T',o$}\\\label{alg:referenceSystemLabelingAlgorithm6}
		Choose an arbitrary cycle $(c_1,c_2,c_3,c_4)\in\mathcal{C}$, where $c_1=o$\\
		$o_1\leftarrow c_2$, $o_2\leftarrow c_4$\\
		\For{$i\leftarrow3$ \KwTo $4$}{
			Find $(c_1,c_2,c_3,c_4)\in\mathcal{C}$, where $c_1=o$, such that $(c_2=o_{i-1}\wedge c_4\ne o_{i-2})\vee(c_4=o_{i-1}\wedge c_2\ne o_{i-2})$\\
			$o_i\leftarrow\{c_2,c_4\}\setminus\{o_{i-1}\}$
		}
		\KwRet $(o_1,o_2,o_3,o_4)$
	}
	\caption{Reference system labeling algorithm}
	\label{alg:referenceSystemLabelingAlgorithm}
\end{algorithm}

The \LRERA (Algorithm~\ref{alg:longRangeEdgesRemovingAlgorithm}) detects a large fraction of the vertices, thus, a large fraction of the long-range edges is removed from $G$.
So, the returned graph $T'=(V,E')$ is an ``almost'' torus.
Next, we select $o\in V$ to become the origin if $o$ and its neighbors in $T'$ have all their long-range edges removed.
\textit{I.e.}, the algorithm sets the vertex $o$ as the origin if \LRERA detected it and its neighbors.
Note that, such $o$ and its neighbors have degree equal to four in $T'$.
Thus, we can label the first cross and set it as the reference system of $G$.

The \emph{reference system labeling algorithm}, \RSLA for short, takes as input the graph $T'$ returned by \LRERA and outputs a queue of vertices $Q$ and a labeling array $\ell$, both already initialized.
Its main purpose is to initialize a breadth-first search that labels the remaining vertices (whenever possible), starting from the labeled vertices of the reference system.
It finds a vertex $o$ that can be the origin, labels the cross rooted in $o$ generating the reference system and enqueues in $Q$ only the neighbors of $o$.
The queue $Q$ is used later to perform the breadth-first search.

\RSLA, defined in Algorithm~\ref{alg:referenceSystemLabelingAlgorithm}, starts assigning the labels of all vertices to {\nil} value and their \emph{enqueue} attribute to \false.
After, it chooses $o\in V$ that can be a potential origin for the reference system, on lines \ref{alg:referenceSystemLabelingAlgorithm1} to \ref{alg:referenceSystemLabelingAlgorithm2}.
If $o$ and its neighbors in $T'$ are detected, then \RSLA sets $o$ as the origin.
It labels $o$ as $(0,0)$ and finds the neighbors of $o$ in $T'$ calling the \emph{reference system vertices finder algorithm}, \RSVFA for short.
\RSVFA runs \CS and adds in $\mathcal{C}$ all four-cycles in $T'$ with root vertex in $o$.
As $o$ is detected, so $\delta_{T'}(o)=4$, and each neighbor of $o$ in the cycles in $\mathcal{C}$ is one of the four vertices (ii) of the reference system.

When the size of the torus is $n=4$, $\mathcal{C}$ has two four-cycles rooted in $o$ composed only by local edges that wrap around the two-dimensional torus $T$.
This causes the labeling algorithm presented in this paper to fail.
Thus, we consider tori of sizes $n\ge5$ in the rest of the paper.
As $n\ge5$ and $o$ and its neighbors have degree four in $T'$, so there are only four distinct cycles in $\mathcal{C}$ and these define a lattice pattern.
So, \RSVFA chooses an arbitrary cycle $(c_1,c_2,c_3,c_4)\in\mathcal{C}$ and sets the first two vertices of the reference system $o_1$ and $o_2$ to $c_2$ and $c_4$, respectively.
This is possible because \CS outputs cycles so that $c_1=o$, $c_2\ne c_4$ and because $\mathcal{C}$ defines a lattice pattern.
\RSVFA finds $o_3$ and $o_4$ in the cycles in $\mathcal{C}$ in a similar way that \LPRA does.

After \RSVFA finding the vertices of the reference system $o_1$, \ldots, $o_4$, \RSLA labels them on line \ref{alg:referenceSystemLabelingAlgorithm3}.
Note that the algorithm computes $n$ as $|V|^{1/2}$.
Finally, it enqueues the neighbors of $o$ in $Q$ and tags the first cross as already enqueued, in order to not enqueue these vertices again in the breadth-first search, performed by the \emph{labeling algorithm} (\LA), described in Section~\ref{subsec:mainAlgorithm}.
Note that it is not necessary to enqueue $o$ because $o$ and its neighbors are already labeled.
Lemma~\ref{lem:numberOfVerticesChosenAsOrigin} bounds the expected running time of lines from \ref{alg:referenceSystemLabelingAlgorithm1} to \ref{alg:referenceSystemLabelingAlgorithm2}.

\begin{lemma}
	Lines \ref{alg:referenceSystemLabelingAlgorithm1} to \ref{alg:referenceSystemLabelingAlgorithm2} of \RSLA run in expected constant time, when receiving an output of \LRERA.
	\label{lem:numberOfVerticesChosenAsOrigin}
\end{lemma}
\begin{proof}
	Let $X$ be the random variable that counts the number of vertices that \RSLA chooses until it finds an origin for the reference system. Note that $X$ is a geometric random variable.
	Let $D_w$ be the event of \LRERA detecting the vertex $w\in V$.
	The probability of a $o\in V$ being an origin is $p=\Pr\left(D_o\cap\bigcap\limits_{v\in N_T(o)}D_v\right)$, where $N_T(o)$ is the set of neighbors of $o$ in the original spanning torus $T$ of $G$ generated by the \UTSW model.
	By De Morgan's laws, $p=1-\Pr\left(\overline{D_o}\cup\bigcup\limits_{v\in N_T(o)}\overline{D_v}\right)$.
	By Lemma~\ref{lem:DuProbability}, $\Pr\left(\overline{D_w}\right)\le4\Pr(E_w)$ for all $w\in V$.
	The result follows by union bound, Theorem~\ref{thm:upperBoundOfEu} proof and by the fact that $X$ is a geometric random variable.
	Then,
	\begin{alignat*}{2}
		\E[X]&<(1-(&&1400/9\ln^{-1}(n/2)+\\
			& &&(650\zeta(3)+2015/32)\ln^{-2}(n/2)+\\
			& &&(480\zeta(3)+160)(\ln n+1)\ln^{-3}(n/2)+\\
			& &&(15\zeta(3)+5)(\ln n+1)^2\ln^{-4}(n/2)))^{-1}.
	\end{alignat*}
\end{proof}

The running time of Algorithm~\ref{alg:referenceSystemLabelingAlgorithm} depends on: (i) the running time of the initializations of labels and \emph{enqueued} attribute of the vertices on lines \ref{alg:referenceSystemLabelingAlgorithm4} to \ref{alg:referenceSystemLabelingAlgorithm5}; (ii) the running time of the lines \ref{alg:referenceSystemLabelingAlgorithm1} to \ref{alg:referenceSystemLabelingAlgorithm2} and (iii) the running time of \RSVFA, which depends on the running time of \CS on line \ref{alg:referenceSystemLabelingAlgorithm6}.
Lemmas~\ref{lem:numberOfVerticesChosenAsOrigin} and \ref{lem:runningTime4CyclesSearch} bound (ii) and (iii), respectively, and (i) is linear on $|V|$.
Combining these, we state Corollary~\ref{cor:runnigTimeOfReferenceSystemLabelingAlgorithm}.

\begin{corollary}
	\RSLA runs in expected linear time on $|V|$, when receiving an output of \LRERA.
	\label{cor:runnigTimeOfReferenceSystemLabelingAlgorithm}
\end{corollary}

\subsection{Labeling arbitrary crosses}

The previous section shows how to label the vertices of the reference system.
After that procedure, we need to label the remaining vertices of the graph $T'=(V,E')$ that \LRERA outputs.
It is possible to label the cross rooted in a detected and labeled vertex $u\in V$ if there is a lattice pattern rooted in $u$ with some properties.
The algorithm of this section performs this labeling.

The \emph{arbitrary cross labeling algorithm}, \ACLA for short, takes as input the graph $T'$, a detected and labeled vertex $u$, a queue of vertices $Q$ and a labeling array $\ell$.
It outputs the queue $Q$ with the cross rooted in $u$ enqueued and the labeling array $\ell$ with the same cross labeled.
The algorithm starts running \CS to find the set $\mathcal{L}$ of all four-cycles rooted in $u$.
We claim that $\mathcal{L}$ has at least one lattice pattern $\mathcal{C}\subseteq\mathcal{L}$ composed by a cycle with detected and labeled vertices such that these vertices can be used to label the cross rooted in $u$.
Such claim holds because the \emph{labeling algorithm} (\LA), defined in Algorithm~\ref{alg:labelingAlgorithm}, runs a breadth-first search, that already labeled some vertices of at least one cycle in $\mathcal{C}$ that can be used in the labeling.
Section~\ref{subsec:mainAlgorithm} shows this with more details in Theorem~\ref{thm:labelingAlgorithmLabelsCrossOfEachEnqueuedVertex}.

Based on that claim, there is a cycle $(c_1,c_2,c_3,c_4)$ in a lattice pattern $\mathcal{C}\subseteq\mathcal{L}$ that provides a \emph{reference}.
A \emph{reference} (do not confuse with \emph{reference system}) is a pair of local edges in $(c_1,c_2,c_3,c_4)$ with labeled endpoints and that are \emph{perpendicular} in the lattice pattern $\mathcal{C}$.
Two local edges are \emph{perpendicular} in $\mathcal{C}$ if they are consecutive in $(c_1,c_2,c_3,c_4)$.
Figure~\ref{fig:PossibleReferencesToLabelTheCross} illustrates these definitions, where $c_1=u$.
So, the algorithm finds a cycle $(c_1,c_2,c_3,c_4)$ with a reference and places the neighbors of $u$ in the lattice pattern $\mathcal{C}$ on the four endpoints of the cross $u_1$, \ldots, $u_4$, in a similar way that \RSLA does.
After that, it labels $u_1$, \ldots, $u_4$ using the topological information in the reference of $(c_1,c_2,c_3,c_4)$.

\begin{figure}
	\centering
	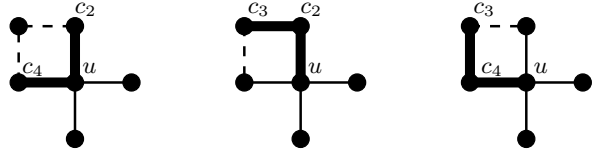
	\caption{Possible references to label the cross. The bold edges represent perpendicular edges.}
	\label{fig:PossibleReferencesToLabelTheCross}
\end{figure}

Let $(c_1,c_2,c_3,c_4)\in\mathcal{C}$ be the four-cycle that has a reference for labeling the cross rooted in $u$.
Recall that $c_1=u$ because \CS outputs $\mathcal{L}$.
If $c_2$ and $c_4$ are labeled, then the edges $\{u,c_2\}$ and $\{u,c_4\}$ are a reference.
These edges are perpendicular in $\mathcal{C}$ and are local edges because $u$ is detected.
Figure~\ref{fig:PossibleReferencesToLabelTheCross} illustrates this case on the left side, where the bold edges are the reference.
Besides that, $(c_1,c_2,c_3,c_4)$ may have three more distinct possible references.
If $c_2$ or $c_3$ are detected, then the edges $\{u,c_2\}$ and $\{c_2,c_3\}$ are local edges because $u$ is detected.
Moreover, if $c_2$ and $c_3$ are labeled, then these edges are a reference, by the fact that $u$ is also labeled.
Figure~\ref{fig:PossibleReferencesToLabelTheCross} illustrates this case in the middle.
The same happens with vertices $c_3$ and $c_4$ such that the edges $\{u,c_4\}$ and $\{c_3,c_4\}$ are a reference.
Figure~\ref{fig:PossibleReferencesToLabelTheCross} illustrates this case in the right side.
It is not necessary to consider the case where the edges $\{c_2,c_3\}$ and $\{c_3,c_4\}$ are a reference.
The reason is that the latter two cases happen simultaneously when this happens, because $u$ is detected and labeled.

Algorithm~\ref{alg:arbitraryCrossLabelingAlgorithm} defines \ACLA.
Lines \ref{alg:arbitraryCrossLabelingAlgorithm1} to \ref{alg:arbitraryCrossLabelingAlgorithm2} find a lattice pattern $\mathcal{C}$ and a reference for labeling the cross rooted in $u$ as explained above.
It sets the first two vertices of the cross, $u_1$ and $u_2$, and labels one of them, if required.
Lines \ref{alg:arbitraryCrossLabelingAlgorithm3} to \ref{alg:arbitraryCrossLabelingAlgorithm4} set the other two vertices, $u_3$ and $u_4$, in a similar way that \RSVFA does.
\ACLA finishes labeling the vertices of the cross, enqueuing in $Q$ those that are detected and are not enqueue yet, and assigning their \emph{enqueued} attribute to {\true} in order to not enqueue those vertices again.

\begin{algorithm}
	\Instance{A graph (``almost'' torus) $T'=(V,E')$, a vertex $u\in V$, a queue $Q$ and a labeling array $\ell$}
	\Returns{A queue $Q$ and the labeling array $\ell$ with the cross rooted in $u$ labeled}
	\BlankLine
	\Alg{\ACLA{$T',u,Q,\ell$}}{
		$\mathcal{L}\leftarrow$ \CS{$T',u$}\\\label{alg:arbitraryCrossLabelingAlgorithm1}
		\ForEach{\emph{lattice pattern} $\mathcal{C}\subseteq\mathcal{L}$}{
			\ForEach{$(c_1,c_2,c_3,c_4)\in\mathcal{C}$\emph{, where} $c_1=u$}{
				\If{$\ell[c_2]\ne$ \emph{\nil} $\wedge$ $\ell[c_4]\ne$ \emph{\nil}}{
					$u_1\leftarrow c_2$, $u_2\leftarrow c_4$\\
					Break both \textbf{for}
				}\ElseIf{$(c_2$.\emph{detected} $\vee$ $c_3$.\emph{detected}$)$ $\wedge$ $\ell[c_2]\ne$ \emph{\nil} $\wedge$ $\ell[c_3]\ne$ \emph{\nil}}{
					$u_1\leftarrow c_2$, $u_2\leftarrow c_4$\\
					$\ell[u_2]\leftarrow\ell[u] + \ell[c_3] - \ell[c_2]$\\\label{alg:arbitraryCrossLabelingAlgorithm5}
					Break both \textbf{for}
				}\ElseIf{$(c_3$.\emph{detected} $\vee$ $c_4$.\emph{detected}$)$ $\wedge$ $\ell[c_3]\ne$ \emph{\nil} $\wedge$ $\ell[c_4]\ne$ \emph{\nil}}{
					$u_1\leftarrow c_2$, $u_2\leftarrow c_4$\\
					$\ell[u_1]\leftarrow\ell[u] + \ell[c_3] - \ell[c_4]$\\\label{alg:arbitraryCrossLabelingAlgorithm6}
					Break both \textbf{for}\label{alg:arbitraryCrossLabelingAlgorithm2}
				}
			}
		}
		\For{$i\leftarrow3$ \KwTo $4$}{\label{alg:arbitraryCrossLabelingAlgorithm3}
			Find $(c_1,c_2,c_3,c_4)\in\mathcal{C}$, where $c_1=u$, such that $(c_2=u_{i-1}\wedge c_4\ne u_{i-2})\vee(c_4=u_{i-1}\wedge c_2\ne u_{i-2})$\\
			$u_i\leftarrow\{c_2,c_4\}\setminus\{u_{i-1}\}$\label{alg:arbitraryCrossLabelingAlgorithm4}
		}
		$\ell[u_3]\leftarrow 2\ell[u] - \ell[u_1]$\\\label{alg:arbitraryCrossLabelingAlgorithm7}
		$\ell[u_4]\leftarrow 2\ell[u] - \ell[u_2]$\\\label{alg:arbitraryCrossLabelingAlgorithm8}
		\For{$i\leftarrow1$ \KwTo $4$}{
			\If{$($\emph{\textbf{not}} $u_i$.\emph{enqueued}$)$ $\wedge$ $u_i$.\emph{detected}}{
				$Q\leftarrow Q\cup\{u_i\}$\\
				$u_i$.enqueued $\leftarrow$ \true
			}
		}
		\KwRet $(Q,\ell)$
	}
	\caption{Arbitrary cross labeling algorithm}
	\label{alg:arbitraryCrossLabelingAlgorithm}
\end{algorithm}

The vertex labeling assignments on lines \ref{alg:arbitraryCrossLabelingAlgorithm5}, \ref{alg:arbitraryCrossLabelingAlgorithm6}, \ref{alg:arbitraryCrossLabelingAlgorithm7} and \ref{alg:arbitraryCrossLabelingAlgorithm8} perform operations of two-dimensional vector addition and subtraction among the labels of the vertices.
These operations may not result in elements of $[\![n]\!]^2$.
The algorithm solves this by replacing each coordinate of the label $\ell[v]_j$ by $\ell[v]_j\mod n$, where $j$ is the label dimension of the vertex $v$.
Combining Lemmas~\ref{lem:runningTime4CyclesSearch} and \ref{lem:sizeOf4CyclesSet}, together with the fact that \LPRA runs in constant time, we conclude the result of the Corollary~\ref{cor:runnigTimeOfArbitraryCrossLabelingAlgorithm}.

\begin{corollary}
	\ACLA runs in expected constant time, when receiving an output of \LRERA.
	\label{cor:runnigTimeOfArbitraryCrossLabelingAlgorithm}
\end{corollary}

\subsection{Main algorithm}
\label{subsec:mainAlgorithm}

This section presents the main procedure of the labeling algorithm.
It removes most of the long-range edges and runs a breadth-first search.
The breadth-first search labels the crosses rooted in a large fraction of the vertices.
The algorithm labels almost all vertices, runs in expected linear time and can be used to define a compact routing scheme for \UTSW graphs.

The \emph{labeling algorithm}, \LA for short, takes as input a \UTSW graph $G=(V,E)$ and outputs an array $\ell$ that represents the labeling function $\ell_T:V\rightarrow[\![n]\!]^2$, where $T$ is the original spanning torus of $G$ generated by \UTSW.
It starts running \LRERA with $G$ as input.
The output is the torus $T'$ with a few remaining long-range edges and with the \emph{detected} attribute of all vertices already assigned.
After that, the algorithm initializes a breadth-first search, finding a vertex that can be the origin, and labels the vertices of the reference system calling \RSLA with $T'$ as input.
The algorithm runs the breadth-first search iteratively dequeing $u$ from the queue $Q$ and calling \ACLA to label the cross rooted in $u$.
\ACLA iteratively labels the vertices of the crosses and enqueues in $Q$ these vertices that were not enqueued yet.
Algorithm~\ref{alg:labelingAlgorithm} defines \LA.

\begin{algorithm}
	\Instance{A \UTSW graph $G=(V,E)$}
	\Returns{An array $\ell$ representing the labeling $\ell_T:V\rightarrow[\![n]\!]^2$}
	\BlankLine
	\Alg{\LA{$G$}}{
		$T'\leftarrow$ \LRERA{$G$}\\
		$(Q,\ell)\leftarrow$ \RSLA{$T'$}\\
		\While{$Q\ne\emptyset$}{\label{alg:labelingAlgorithm1}
			Dequeue $u$ from $Q$\\
			$(Q,\ell)\leftarrow$ \ACLA{$T',u,Q,\ell$}
		}
		\KwRet $\ell$
	}
	\caption{Labeling algorithm}
	\label{alg:labelingAlgorithm}
\end{algorithm}

\LA may not label some vertices of $G$, so $\ell$ defines a \emph{partial} two-dimensional toroidal vertex labeling function.
However, it labels most of them, because \LRERA detects most of the vertices, as Lemma~\ref{lem:numberOfDetectedVertices} shows.
Besides that, some of the non-detected vertices are labeled in the end of \LA running.
This happens because if one of the four neighbors $w\in N_T(v)$ of a non-detected vertex $v\in V$ in the original spanning torus $T$, was enqueued in $Q$, then \ACLA labels $v$, that are in the cross rooted in $w$.

Theorem~\ref{thm:labelingAlgorithmLabelsCrossOfEachEnqueuedVertex} shows that \LA labels the crosses rooted in all enqueued vertices.
The breadth-first search runs only over the local edges, in consequence of visiting only detected vertices.
Then, \LA does not label a vertex $v\in V$ only if $v$ is not reachable from the origin of the reference system $o\in V$ by a path composed only by local edges and detected vertices.
That is, $v$ is inside an area of $T$ in which the boundaries are composed by non-detected vertices only.

\begin{theorem}
	\LA labels all the vertices of crosses rooted in each vertex of $Q$.
	\label{thm:labelingAlgorithmLabelsCrossOfEachEnqueuedVertex}
\end{theorem}
\begin{proof}
	Let $u_i$ be the vertex that is in the front of $Q$ in the $i^{\textrm{th}}$ iteration of line \ref{alg:labelingAlgorithm1} in Algorithm~\ref{alg:labelingAlgorithm}.
	Let $s_i$ be the root of the cross processed by \RSLA or \ACLA when $u_i$ was enqueued in $Q$.
	Note that $s_i$ and $u_i$ are both detected, because \RSLA chooses a detected origin and enqueues detected vertices, and \ACLA enqueues only detected vertices.
	The edge $\{u_i,s_i\}$ belongs to the original spanning torus $T$, because both endpoints are detected.
	Let $c_{i1}=(u_i,s_i,v_1,w_1)$ and $c_{i2}=(u_i,s_i,v_2,w_2)$ be the two distinct four-cycles induced in $T$.
	The vertices $u_i$, $v_1$ and $v_2$ are all labeled, because either \RSLA already labeled them if $s_i$ is the \emph{origin}, or \ACLA already labeled them during the labeling of the cross in $s_i$.
	Also, the vertex $s_i$ is also labeled, because either \RSLA already labeled $s_i$ if it is the \emph{origin}, or \RSLA or \ACLA already labeled it before enqueuing it in $Q$.
	Without loss of generality, $\{u_i,s_i\}$ and $\{s_i,v_1\}$ have labeled endpoints and are local edges perpendicular in the lattice pattern rooted in $u_i$ and induced in $T$, \textit{i.e.}, they are a reference in $c_{i1}$.
	Let $\mathcal{L}_{u_i}$ be the list assigned on line \ref{alg:arbitraryCrossLabelingAlgorithm1} of \ACLA when $u=u_i$.
	Note that $c_{i1}$ are in $\mathcal{L}_{u_i}$.
	As $c_{i1}$ has the two edges $\{u_i,s_i\}$ and $\{s_i,v_1\}$ as a reference and $s_i$ is detected, so \ACLA labels the vertices of the cross in $u_i$.
	Therefore, \LA labels the cross of each vertex $u_i$ enqueued in $Q$.
\end{proof}

Besides some non-labeled vertices may exist, this number is small.
Corollary~\ref{cor:allVerticesAreDetectedForInfiniteN} shows that the number of detected vertices tends to $|V|$ whereas the size of the torus $n$ grows.
Equivalently, the number of non-detected vertices tends to zero.
The combination of this fact with Theorem~\ref{thm:labelingAlgorithmLabelsCrossOfEachEnqueuedVertex} results in Theorem~\ref{thm:numberVerticesLabelingAlgorithmLabelsTendsToN} statement.

\begin{theorem}
	As $n$ goes to infinity, the number of vertices that \LA labels tends to $|V|$.
	\label{thm:numberVerticesLabelingAlgorithmLabelsTendsToN}
\end{theorem}

Algorithm~\ref{alg:labelingAlgorithm} runs \LRERA and \RSLA once, demanding expected linear time, by Lemma~\ref{lem:runningTimeLongRangeEdgesRemovingAlgorithm} and Corollary~\ref{cor:runnigTimeOfReferenceSystemLabelingAlgorithm}.
Each iteration of the breadth-first search runs \ACLA once, demanding expected constant time, by Corollary~\ref{cor:runnigTimeOfArbitraryCrossLabelingAlgorithm}.
As the breadth-first search runs at most $|V|-1$ iterations, then the \LA runs in expected linear time, as stated by Theorem~\ref{thm:runnigTimeOfLabelingAlgorithm}.

\begin{theorem}
	\LA runs in expected linear time on $|V|$.
	\label{thm:runnigTimeOfLabelingAlgorithm}
\end{theorem}

The main application of \LA is in routing of messages in \UTSW graphs.
Section~\ref{sec:toroidalSmallWorldLabelingProblem} claims that myopic search can route messages if a two-dimensional toroidal vertex labeling function $\ell$ (Definition~\ref{dfn:2DimensionalToroidalVertexLabelingFunction}) is known.
In this sense, each vertex $u\in V$ requires $2\lceil\log n\rceil$ bits for its label $\ell[u]$, $\lceil\log n\rceil$ bits for each dimension, where $n$ is the size of the torus.
Also, $u$ requires expected O$(\log n)$ bits for its routing table, because each neighbor $v\in N_G(u)$ of $u$ in $G$ has a row in $u$'s routing table, encoded by the pair $(\ell[v],p_u(v))$, where $p_u(v)$ is the logical port id in $u$ that directs the message to the edge in $u$ incident to $v$ in $G$.
As $\ell[v]$ requires $2\lceil\log n\rceil$ bits, $p_u(v)$ requires expected O$(1)$ bits and $u$'s routing table has expected O$(1)$ rows, both by Lemma~\ref{lem:expectedDegree}, so $u$'s routing table has expected O$(\log n)$ bits.
Then, each $u\in V$ requires expected O$(\log n)$ bits for running the myopic search, which is sub-linear in the size of the network $G$, given $|V|=n^2$ and $|E|\le3|V|$.

Considering the case of generating the routing tables with a partial two-dimensional toroidal vertex labeling function, there is a compact routing scheme for $G$ with preprocessing algorithm that generates sub-linear structures per vertex and with routing algorithm that forwards messages in expected constant time.
Theorem~\ref{thm:compactRoutingSchemeForUTSWGraphs} shows this.

\begin{theorem}
	There is a compact routing scheme for \UTSW graphs.
	\label{thm:compactRoutingSchemeForUTSWGraphs}
\end{theorem}
\begin{proof}
	The following preprocessing algorithm generates the labels and routing tables for each vertex of $G$.
	Run \LA with input $G$, generating the labeling $\ell$.
	For each $u\in V$, assign $\ell[u]$ to its label.
	For each $\{u,v\}\in E$, add $(\ell[v],p_u(v))$ in $u$'s routing table and $(\ell[u],p_v(u))$ in $v$'s routing table.
	\LA generates the labeling $\ell$ in expected linear time, as Theorem~\ref{thm:runnigTimeOfLabelingAlgorithm} states.
	The preprocessing algorithm generates the labels and the routing tables in $\Theta(|V|)$ time, because each access in $\ell$ and in $p$ runs in constant time and because $2|V|\le|E|\le3|V|$.
	So, the preprocessing algorithm runs in O$(|V|)$ expected time and generates structures with expected O$(\log n)$ bits for each vertex, as explained above.
	The myopic search runs in expected constant time, by Lemma~\ref{lem:expectedDegree}.
	Therefore, there is an expected linear time preprocessing algorithm that generates structures with expected sub-linear size for each vertex and a related routing algorithm that runs in expected constant time on each vertex.
\end{proof}

\section{Conclusion}
\label{sec:conclusion}

In this work, we present a small world graph model, called \UTSW model, and a labeling algorithm for this model.
The model topology, the linear time execution of the labeling algorithm, and the labels itself, imply in a compact routing scheme for the graph.

The \UTSW model is built upon a two-dimensional torus together with random long-range edges.
The resulting graph has average distance O$(\log n)$, and allows the myopic search to perform O$\left(\log^2n\right)$ expected forwards for a message reach its destination.
Therefore, a \UTSW graph exhibits the two main properties of the small world networks: clustering and paths with small sizes.

The structure of the underlying torus can be seen as a well-formed pattern of four-cycles, which motivates the approach used to label the vertices.
The difficulty arises from the random edges, which may create new four-cycles that do not belong to the torus.
Indeed, in extreme cases (see Figure~\ref{fig:Ambiguity}), a unique labeling is not even possible.
Nevertheless, for large graphs, the proposed algorithm labels almost all vertices (Theorem~\ref{thm:numberVerticesLabelingAlgorithmLabelsTendsToN}).

Generally speaking, labeling random graphs poses serious challenges.
For future works, one can consider graphs that better models the real world.
For example, the torus used here models a situation where the vertices are equally distributed on the geographical space.
Some models \cite{Liben-Nowell:2005,Kleinberg:2001} deal with this issue, but it remains an open problem if the results here can be adapted to them.
Also, several real-world network models can be considered, for example, models for power law graphs, and the hyperbolic geometric graph.
\acknowledgements

We thank the Coordination for the Improvement of Higher Education Personnel (CAPES) by the scholarship financial support.

\nocite{*}
\bibliographystyle{abbrvnat}
% use the following instead if you encounter problems 
%\bibliographystyle{alpha}
\bibliography{dmtcs-bibliography}

\begin{thebibliography}{33}
\providecommand{\natexlab}[1]{#1}
\providecommand{\url}[1]{\texttt{#1}}
\expandafter\ifx\csname urlstyle\endcsname\relax
  \providecommand{\doi}[1]{doi: #1}\else
  \providecommand{\doi}{doi: \begingroup \urlstyle{rm}\Url}\fi

\bibitem[Abraham et~al.(2006)Abraham, Gavoille, Goldberg, and
  Malkhi]{Abraham:2006a}
I.~Abraham, C.~Gavoille, A.~V. Goldberg, and D.~Malkhi.
\newblock {Routing in Networks with Low Doubling Dimension}.
\newblock In \emph{26$^{th}$ IEEE International Conference on Distributed
  Computing Systems}, page~75, Lisbon - Portugal, July 2006.

\bibitem[Adamic and Adar(2005)]{Adamic:2005}
L.~Adamic and E.~Adar.
\newblock {How to search a social network}.
\newblock \emph{Social Networks}, 27\penalty0 (3):\penalty0 187--203, July
  2005.

\bibitem[Aspnes et~al.(2002)Aspnes, Diamadi, and Shah]{Aspnes:2002}
J.~Aspnes, Z.~Diamadi, and G.~Shah.
\newblock {Fault-tolerant Routing in Peer-to-peer Systems}.
\newblock In \emph{Twenty-first Annual Symposium on Principles of Distributed
  Computing}, pages 223--232, Monterey, California - United States of America,
  July 2002.

\bibitem[Ban et~al.(2010)Ban, Gao, and van~de Rijt]{Ban:2010}
X.~Ban, J.~Gao, and A.~van~de Rijt.
\newblock {Navigation in Real-World Complex Networks through Embedding in
  Latent Spaces}.
\newblock In \emph{Twelfth Workshop on Algorithm Engineering \& Experiments},
  pages 138--148, Austin, Texas - United States of America, January 2010.

\bibitem[Bollob\'as and Chung(1988)]{Bollobas:1988}
B.~Bollob\'as and F.~R.~K. Chung.
\newblock {The Diameter of a Cycle Plus a Random Matching}.
\newblock \emph{SIAM Journal on Discrete Mathematics}, 1\penalty0 (3):\penalty0
  328--333, January 1988.

\bibitem[Brady and Cowen(2006)]{Brady:2006}
A.~Brady and L.~Cowen.
\newblock {Compact Routing on Power Law Graphs with Additive Stretch}.
\newblock In \emph{Eighth Workshop on Algorithm Engineering \& Experiments},
  pages 119--128, Miami, Florida - United States of America, January 2006.

\bibitem[Bringmann et~al.(2017)Bringmann, Keusch, Lengler, Maus, and
  Molla]{Bringmann:2017}
K.~Bringmann, R.~Keusch, J.~Lengler, Y.~Maus, and A.~R. Molla.
\newblock {Greedy Routing and the Algorithmic Small-World Phenomenon}.
\newblock In \emph{ACM Symposium on Principles of Distributed Computing}, pages
  371--380, Washington, District of Columbia - United States of America, July
  2017.

\bibitem[Chen et~al.(2009)Chen, Sommer, Teng, and Wang]{Chen:2009}
W.~Chen, C.~Sommer, S.-H. Teng, and Y.~Wang.
\newblock {Compact Routing in Power-Law Graphs}.
\newblock In \emph{23$^{rd}$ International Conference on Distributed
  Computing}, pages 379--391, Elche, Alicante - Spain, September 2009.

\bibitem[Chen et~al.(2012)Chen, Sommer, Teng, and Wang]{Chen:2012}
W.~Chen, C.~Sommer, S.-H. Teng, and Y.~Wang.
\newblock {A Compact Routing Scheme and Approximate Distance Oracle for
  Power-law Graphs}.
\newblock \emph{ACM Transactions on Algorithms}, 9:\penalty0 4:1--4:26, 2012.
\newblock \doi{10.1145/2390176.2390180}.

\bibitem[Clarke et~al.(2000)Clarke, Sandberg, Wiley, and Hong]{Clarke:2000}
I.~Clarke, O.~Sandberg, B.~Wiley, and T.~W. Hong.
\newblock {Freenet: A Distributed Anonymous Information Storage and Retrieval
  System}.
\newblock In \emph{ICSI Workshop on Design Issues in Anonymity and
  Unobservability}, pages 311--320, Berkeley, California - United States of
  America, July 2000.

\bibitem[Cowen(1999)]{Cowen:1999}
L.~J. Cowen.
\newblock {Compact Routing with Minimum Stretch}.
\newblock In \emph{Tenth Annual ACM-SIAM Symposium on Discrete Algorithms},
  pages 255--260, Baltimore, Maryland - United States of America, January 1999.

\bibitem[Dietzfelbinger and Woelfel(2009)]{Dietzfelbinger:2009}
M.~Dietzfelbinger and P.~Woelfel.
\newblock {Tight Lower Bounds for Greedy Routing in Uniform Small World Rings}.
\newblock In \emph{Forty-first Annual ACM Symposium on Theory of Computing},
  pages 591--600, Bethesda, Maryland - United States of America, June 2009.

\bibitem[Easley and Kleinberg(2010)]{Easley:2010}
D.~Easley and J.~Kleinberg.
\newblock \emph{{Networks, Crowds, and Markets: Reasoning About a Highly
  Connected World}}, chapter The Small-World Phenomenon, pages 611--644.
\newblock Cambridge University Press, 2010.

\bibitem[Fraigniaud et~al.(2006)Fraigniaud, Gavoille, and
  Paul]{Fraigniaud:2006}
P.~Fraigniaud, C.~Gavoille, and C.~Paul.
\newblock {Eclecticism shrinks even small worlds}.
\newblock \emph{Distributed Computing}, 18\penalty0 (4):\penalty0 279--291,
  March 2006.

\bibitem[Jovanovi\'{c}(2001)]{Jovanovic:2001}
M.~Jovanovi\'{c}.
\newblock {Modeling peer-to-peer network topologies through ``small-world''
  models and power laws}, 2001.

\bibitem[Kleinberg(2000)]{Kleinberg:2000a}
J.~Kleinberg.
\newblock {The Small-World Phenomenon: An Algorithmic Perspective}.
\newblock In \emph{Thirty-second Annual ACM Symposium on Theory of Computing},
  pages 163--170, Portland, Oregon - United States of America, May 2000.

\bibitem[Kleinberg(2001)]{Kleinberg:2001}
J.~Kleinberg.
\newblock {Small-World Phenomena and the Dynamics of Information}.
\newblock In \emph{14$^{th}$ International Conference on Neural Information
  Processing Systems: Natural and Synthetic}, pages 431--438, Vancouver,
  British Columbia - Canada, December 2001.

\bibitem[Kleinberg(2006)]{Kleinberg:2006}
J.~Kleinberg.
\newblock {Complex networks and decentralized search algorithms}.
\newblock In \emph{International Congress of Mathematicians}, pages 1019--1044,
  Madrid - Spain, August 2006.

\bibitem[Konjevod et~al.(2007)Konjevod, Richa, and Xia]{Konjevod:2007a}
G.~Konjevod, A.~W. Richa, and D.~Xia.
\newblock {Optimal Scale-free Compact Routing Schemes in Networks of Low
  Doubling Dimension}.
\newblock In \emph{Eighteenth Annual ACM-SIAM Symposium on Discrete
  Algorithms}, pages 939--948, New Orleans, Louisiana - United States of
  America, January 2007.

\bibitem[Krioukov et~al.(2004)Krioukov, Fall, and Yang]{Krioukov:2004}
D.~Krioukov, K.~Fall, and X.~Yang.
\newblock {Compact Routing on Internet-Like Graphs}.
\newblock In \emph{23$^{rd}$ IEEE Conference on Computer Communications}, pages
  209--219, Hong Kong, Pearl River - China, March 2004.

\bibitem[Krioukov et~al.(2008)Krioukov, Papadopoulos, Bogu{\~n}{\'a}, and
  Vahdat]{Krioukov:2008}
D.~Krioukov, F.~Papadopoulos, M.~Bogu{\~n}{\'a}, and A.~Vahdat.
\newblock {Efficient Navigation in Scale-Free Networks Embedded in Hyperbolic
  Metric Spaces}.
\newblock Technical report, Center for Applied Internet Data Analysis, 2008.
\newblock URL \url{http://cds.cern.ch/record/1103993}.
\newblock Accessed in 12/11/2017.

\bibitem[Krioukov et~al.(2009)Krioukov, Papadopoulos, Bogu{\~n}{\'a}, and
  Vahdat]{Krioukov:2009}
D.~Krioukov, F.~Papadopoulos, M.~Bogu{\~n}{\'a}, and A.~Vahdat.
\newblock {Greedy Forwarding in Scale-Free Networks Embedded in Hyperbolic
  Metric Spaces}.
\newblock \emph{ACM SIGMETRICS Performance Evaluation Review}, 37\penalty0
  (2):\penalty0 15--17, September 2009.

\bibitem[Liben-Nowell et~al.(2005)Liben-Nowell, Novak, Kumar, Raghavan, and
  Tomkins]{Liben-Nowell:2005}
D.~Liben-Nowell, J.~Novak, R.~Kumar, P.~Raghavan, and A.~Tomkins.
\newblock {Geographic routing in social networks}.
\newblock \emph{Proceedings of the National Academy of Sciences}, 102\penalty0
  (33):\penalty0 11623--11628, August 2005.

\bibitem[Liu et~al.(2009)Liu, Guan, Bai, and Lu]{Liu:2009}
X.~Liu, J.~Guan, G.~Bai, and H.~Lu.
\newblock {SWER: small world-based efficient routing for wireless sensor
  networks with mobile sinks}.
\newblock \emph{Frontiers of Computer Science in China}, 3\penalty0
  (3):\penalty0 427--434, September 2009.

\bibitem[Manku et~al.(2003)Manku, Bawa, and Raghavan]{Manku:2003}
G.~S. Manku, M.~Bawa, and P.~Raghavan.
\newblock {Symphony: Distributed Hashing In A Small World}.
\newblock In \emph{4$^{th}$ USENIX Symposium on Internet Technologies and
  Systems}, pages 127--140, Seattle, Washington - United States of America,
  March 2003.

\bibitem[Manku et~al.(2004)Manku, Naor, and Wieder]{Manku:2004}
G.~S. Manku, M.~Naor, and U.~Wieder.
\newblock {Know thy Neighbor's Neighbor: the Power of Lookahead in Randomized
  P2P Networks}.
\newblock In \emph{Thirty-sixth Annual ACM Symposium on Theory of Computing},
  pages 54--63, Chicago, Illinois - United States of America, June 2004.

\bibitem[Martel and Nguyen(2004)]{Martel:2004}
C.~Martel and V.~Nguyen.
\newblock {Analyzing Kleinberg's (and other) Small-world Models}.
\newblock In \emph{Twenty-third Annual ACM Symposium on Principles of
  Distributed Computing}, pages 179--188, St. John's, Newfoundland and Labrador
  - Canada, July 2004.

\bibitem[Milgram(1967)]{Milgram:1967}
S.~Milgram.
\newblock {The Small World Problem}.
\newblock \emph{Psychology Today}, 1:\penalty0 61--67, 1967.

\bibitem[Sandberg(2006)]{Sandberg:2006}
O.~Sandberg.
\newblock {Distributed Routing in Small-World Networks}.
\newblock In \emph{Eighth Workshop on Algorithm Engineering \& Experiments},
  pages 144--155, Miami, Florida - United States of America, January 2006.

\bibitem[Thorup and Zwick(2001)]{Thorup:2001a}
M.~Thorup and U.~Zwick.
\newblock {Compact routing schemes}.
\newblock In \emph{Thirteenth Annual ACM Symposium on Parallel Algorithms and
  Architectures}, pages 1--10, Heraklion, Crete - Greece, July 2001.

\bibitem[Watts and Strogatz(1998)]{Watts:1998}
D.~J. Watts and S.~H. Strogatz.
\newblock {Collective dynamics of `small-world' networks}.
\newblock \emph{Nature}, 393\penalty0 (6684):\penalty0 440--442, June 1998.

\bibitem[Zeng et~al.(2005)Zeng, Hsu, and Wang]{Zeng:2005}
J.~Zeng, W.-J. Hsu, and J.~Wang.
\newblock {Near Optimal Routing in a Small-World Network with Augmented Local
  Awareness}.
\newblock In \emph{Third International Conference on Parallel and Distributed
  Processing and Applications}, pages 503--513, Nanjing, Jiangsu - China,
  November 2005.

\bibitem[Zeng and Hsu(2006)]{Zeng:2006}
J.-Y. Zeng and W.-J. Hsu.
\newblock {Optimal Routing in a Small-World Network}.
\newblock \emph{Journal of Computer Science and Technology}, 21\penalty0
  (4):\penalty0 476--481, July 2006.

\end{thebibliography}

\end{document}